\documentclass[a4paper,10pt]{article}

\usepackage{geometry} 

\usepackage{comment}

\usepackage[T1]{fontenc}
\usepackage[utf8]{inputenc}

\usepackage{amssymb, amsthm, amsmath}

\usepackage{authblk}
\usepackage[hidelinks]{hyperref} 
\usepackage[usenames,dvipsnames]{xcolor}
\usepackage[biblabel]{cite}

\usepackage{comment}
\usepackage{mathdots}

\usepackage{bm,xcolor,tikz,hyperref, float}

\newtheorem{theorem}{Theorem}[section]
\newtheorem{lemma}[theorem]{Lemma}
\newtheorem{proposition}[theorem]{Proposition}
\newtheorem{corollary}[theorem]{Corollary}

\theoremstyle{definition}
\newtheorem{definition}[theorem]{Definition} 
\newtheorem{remark}[theorem]{Remark}
\newtheorem{example}[theorem]{Example}

\newcommand{\w}{{\rm w}}
\newcommand{\dd}{{\rm d}}
\newcommand{\C}{\mathcal{C}}
\newcommand{\inn}{\mathrm{in}}
\newcommand{\F}{\mathbb{F}}
\newcommand{\xx}{\mathbf{x}}
\newcommand{\yy}{\mathbf{y}}
\newcommand{\aaa}{\mathbf{a}}
\newcommand{\bb}{\mathbf{b}}
\newcommand{\ii}{\mathbf{i}}
\newcommand{\jj}{\mathbf{j}}
\newcommand{\kk}{\mathbf{k}}
\newcommand{\cc}{\mathbf{c}}
\newcommand{\rr}{\mathbf{r}}
\newcommand{\uu}{\mathbf{u}}
\newcommand{\vv}{\mathbf{v}}
\newcommand{\ones}{\mathbf{1}}
\newcommand{\J}{\mathcal{J}}
\newcommand{\coef}{{\rm Coeff}}

\DeclareMathOperator{\ini}{in} 

\usepackage{mathtools} 
\DeclarePairedDelimiter\abs{\lvert}{\rvert}%
\DeclarePairedDelimiter\norm{\lVert}{\rVert}%

\makeatletter
\let\oldabs\abs
\def\abs{\@ifstar{\oldabs}{\oldabs*}}
\let\oldnorm\norm
\def\norm{\@ifstar{\oldnorm}{\oldnorm*}}
\makeatother

\newcommand{\umberto}[1]{{\color{RoyalPurple} \sf $\star\star$ Umberto: [#1]}}

\newcommand{\rmv}[1]{}
\newcommand{\rodrigo}[1]{{\color{Magenta} \sf $\star\star$ Rodrigo: [#1]}}

\title{Duals of multiplicity codes}
\author[1]{Eduardo Camps Moreno}
\author[2]{Adri{\'a}n Fidalgo-D{\'i}az}
\author[1]{Hiram H. L{\'o}pez}
\author[2]{Umberto Mart{\'i}nez-Pe\~{n}as}
\author[2]{Diego Ruano}
\author[2]{Rodrigo San-Jos{\'e}}
\affil[1]{Department of Mathematics\\Virginia Tech, VA, USA\footnote{Emails: \{eduardoc, hhlopez\}@vt.edu}}
\affil[2]{IMUVa-Mathematics Research Institute\\University of Valladolid, Spain\footnote{Emails: \{adrian.fidalgo22, umberto.martinez, diego.ruano, rodrigo.san-jose\}@uva.es}}

\date{}

\begin{document}
\maketitle
\begin{abstract}
Multivariate multiplicity codes have been recently explored because of their importance for list decoding and local decoding.  Given a multivariate multiplicity code, in this paper, we compute its dimension using Gr{\"o}bner basis tools, its dual in terms of indicator functions, and explicitly describe a parity-check matrix. In contrast with Reed--Muller, Reed--Solomon, univariate multiplicity, and other evaluation codes, the dual of a multivariate multiplicity code is not equivalent or isometric to a multiplicity code (i.e., this code family is not closed under duality). We use our explicit description to provide a lower bound on the minimum distance for the dual of a multiplicity code.

\textbf{Keywords:} Evaluation codes, footprint bound, multiplicity codes, polynomial codes, polynomial ideal codes, Reed--Muller codes, Reed--Solomon codes, Schwartz-Zippel bound.

\end{abstract}

\section{Introduction} \label{sec intro}

Reed-Solomon and Reed--Muller codes are well-known families of linear error-correcting codes obtained by evaluating polynomials over a finite field. They are the most widely used families of algebraic codes, and many generalizations have been proposed over the years, such as folded Reed-Solomon codes \cite{guruswamiFoldedRSDef} or polynomial ideal codes \cite{guruswamiPolynomialIdealDef}. An interesting generalization is given by multiplicity codes, which were introduced in \cite{koppartyMultiplicityDef} and generalize both Reed-Solomon and Reed--Muller codes by also evaluating the derivatives of the polynomials. The first motivation to study these codes was that they provide a family that can be locally decodable efficiently with an asymptotic rate of $1$ \cite{koppartyMultiplicityDef}. Moreover, they can also be list decoded, achieving the optimal rate versus list-decodability tradeoff \cite{guruswami2013linear, koppartyListDecodingMultiplicity}. 

Duals and parity-check matrices of any family of codes are ubiquitous in many applications of coding theory: they are essential for syndrome-based decoding algorithms \cite{sakataplusvoting,duursmaMajority,fengraoMajority} and are required to construct secret sharing schemes \cite{chen2007secure} or quantum codes \cite{kkks,ballQECandtheirGeometries}, among others. Furthermore, in the case of Reed--Muller codes (multiplicity $ 1 $), their duals are used to estimate their weight enumerator and in one of the proofs that they achieve the Shannon capacity on different channels \cite[Sec. 4.3]{abbe2023reed}. 

More recently, the authors of~\cite{woottersListDecodingPolynomialIdeal} proved that the dual of a univariate multiplicity code is of the same type, modulo an isometry. 
Even though \cite{woottersListDecodingPolynomialIdeal} is currently withdrawn from the arXiv server, we keep the reference because we learned about univariate multiplicity codes from there. 

In this work, we explicitly compute the duals (and parity-check matrices) of general (multivariate) multiplicity codes. For multiplicity $ 1 $ (i.e., Reed--Solomon and Reed--Muller codes) the duals belong to the same family (modulo an isometry), and thus their parity-check matrices can be easily described. This holds even for more general families of (multiplicity $ 1 $) multivariate evaluation codes \cite{lopez2021dual}. For higher multiplicities, similar results can be obtained for the case of $1$ variable (i.e., univariate multiplicity codes, see \cite{woottersListDecodingPolynomialIdeal}). Surprisingly, we show that this ``closedness under duality'' does no longer hold for multiplicity $ > 1 $ and more than $ 1 $ variable (see, e.g., Example \ref{ex:dualnoesmult} and Theorem \ref{th dual multiplicity code}), revealing an unexpected behaviour of duals of multiplicity codes, compared to previous evaluation codes.

We use our explicit description of duals of multiplicity codes to lower bound their minimum distance (see Proposition~\ref{prop lower bound dist mult SZ dual}).

The organization is as follows. In Section \ref{sec preliminaries}, we give preliminaries on Hasse derivatives and multiplicities.  In Section \ref{sec hermite}, we construct Hermite interpolation polynomials, whose coefficients will be used to express duals. In Section \ref{sec dimension}, we compute the dimensions of multiplicity codes, which will be used for duality. In Section \ref{sec dual}, we explicitly compute duals of multiplicity codes. For the convenience of the reader, we added Subsections \ref{subsec m=1 dual} and \ref{subsec duals m=r=2}, where we explicitly compute duals for the univariate case (Proposition \ref{p:dualm=1}) and the case of $ 2 $ variables and multiplicity $ 2 $ (Corollary \ref{cor dual multi m=2} and Proposition \ref{p:basedualm=2}), respectively. Then in Subsection \ref{subsec duals general}, we explicitly compute duals in general (Theorems \ref{th dual multiplicity code} and \ref{t:dualcomoeval}). 
At the end of Subsection \ref{subsec duals general}, we give a lower bound on the minimum distance of duals of multivariate multiplicity codes (Proposition \ref{prop lower bound dist mult SZ dual}).

\section{Preliminaries} \label{sec preliminaries}
In this section, we introduce notation and the main relevant results from the literature.

\textbf{Basic notation.} Throughout this manuscript, $ \F $ denotes an arbitrary field and $ \F_q $ the finite field of size $ q $, which is a power of a prime. We denote by $ \F[\xx] = \F[x_1, \ldots, x_m] $ the ring of polynomials in the $ m $ variables $ x_1, \ldots, x_m $ with coefficients in $ \F $, and we denote by $ \F[\xx]_{< k} $ and $ \F[\xx]_{\leq k} $ the sets of polynomials in $ \F[\xx] $ of degree less than $ k $ and at most $ k $, respectively. We set $ \mathbb{N} = \{ 0,1,2, \ldots \} $, and for $ \ii = (i_1, \ldots, i_m) \in \mathbb{N}^m $, we denote $ \xx^\ii = x_1^{i_1} \cdots x_m^{i_m} $ and $ |\ii| = i_1 + \cdots + i_m $. Note that $\deg(\xx^\ii) = |\ii|$. For a monomial ordering $ \prec $, we denote by $ \inn_{\prec}(f) $ (or just $ \inn(f) $) the initial or leading monomial of $ f \in \F[\xx] $. We denote by $ \coef(f,\xx^\ii) \in \F $ the coefficient of $ f \in \F[\xx] $ in the monomial $ \xx^\ii $. We denote by $ \langle A \rangle $ the ideal generated by $ A $ in a ring, and $ \langle A \rangle_\F $ the vector space over $ \F $ generated by $ A $. For integers $ m \leq n $, we denote $ [m,n] = \{ m,m+1, \ldots, n\} $ and $ [n] = [1,n] $. 

For a code $ \C \subseteq \F^{tn} $, we will consider the $ t $-folded Hamming metric in $ \F^{tn} $, i.e., the Hamming metric of length $ n $ over the alphabet $ \F^t $. More precisely, we define
$$ \w (\cc_1, \ldots, \cc_n) = |\{ i \in [n] : \cc_i \neq \mathbf{0} \}|, $$
for $ \cc_1, \ldots , \cc_n \in \F^t $. We then define the folded distance as $ \dd(\cc,\cc^\prime) = \w(\cc - \cc^\prime) $, for $ \cc,\cc^\prime \in \F^{tn} $, and similarly for the minimum folded distance $ \dd(\C) $ of a code $ \C \subseteq \F^{tn} $.

For a code $ \C \subseteq \F^{tn} $, we define the classical Euclidean dual as \[ \C^\perp = \{ \cc \in \F^{tn} : \cc \cdot \cc^\prime = 0, \textrm{ for all } \cc^\prime \in \C \},\]
where
$\cc \cdot \cc^\prime = (c_1,\ldots,c_{tn}) \cdot (c^\prime_1, \ldots, c^\prime_{tn}) = \sum_{i=1}^{tn} c_ic^\prime_i $.

\subsection{Multiplicity codes}
In this work, we consider Hasse derivatives to count multiplicities. 

\begin{definition}[\textbf{Hasse derivative} \cite{hasse}] \label{def Hasse derivative}
Take $ f(\xx) \in \F[\xx] $. Given another family of independent variables $ \yy = (y_1, \ldots, y_m) $, the polynomial $ f(\xx + \yy) \in \F[\xx,\yy] $ can be written uniquely as
$$ f(\xx + \yy) = \sum_{\ii \in \mathbb{N}^m} f^{(\ii)}(\xx) \yy^\ii, $$
for some polynomials $ f^{(\ii)}(\xx) \in \F[\xx] $, for $ \ii \in \mathbb{N}^m $. For $ \ii \in \mathbb{N}^m $, we define the $ \ii $th Hasse derivative of $ f(\xx) $ as the polynomial $ f^{(\ii)}(\xx) \in \F[\xx] $.
\end{definition}

Using the Taylor expansion of a polynomial, it is immediate to see that the classic $ \ii $th derivative of $ f(\xx) $ is given in terms of the Hasse derivative by
$$ \frac{\partial^{|\ii|}}{\partial x_1^{i_1} \cdots \partial x_m^{i_m}} f(\xx) = \ii ! f^{(\ii)}(\xx), $$
where $ \ii! = i_1 ! \cdots i_m! $. Therefore, in positive characteristic, classical derivatives can be recovered from Hasse derivatives, but not vice versa. Hence, Hasse derivatives contain (strictly) more information about the polynomial (in characteristic $ 0 $, they contain the same information). If $ \ii \in [0,1]^m $, then both derivatives coincide over any field since $ \ii ! = 1 $. We now present the Leibniz rule for Hasse derivatives.

\begin{proposition} [\textbf{Leibniz rule}] \label{prop leibniz}
For $ f,g \in \F[\xx] $, we have
$$ (fg)^{(\ii)}(\xx) = \sum_{\jj + \kk = \ii} f^{(\jj)}(\xx) g^{(\kk)}(\xx) . $$
\end{proposition}
\begin{proof}
By the definition of Hasse derivatives,
\begin{equation*}
\begin{split}
(fg)(\xx+\yy) & = \left( \sum_{\jj \in \mathbb{N}^m} f^{(\jj)}(\xx) \yy^\jj \right) \left( \sum_{\kk \in \mathbb{N}^m} g^{(\kk)}(\xx) \yy^\kk \right) 
  = \sum_{\ii \in \mathbb{N}^m} \left( \sum_{\jj + \kk = \ii} f^{(\jj)}(\xx) g^{(\kk)}(\xx) \right) \yy^\ii .
\end{split}
\end{equation*}
Thus, we obtain the result.
\end{proof}

Multiplicity codes were introduced in \cite[Def. 3.4]{koppartyMultiplicityDef}. We will consider them for evaluation points in a grid or Cartesian product $ S = S_1 \times \cdots \times S_m $, where every $ S_i \subseteq \F $ is finite. In this work, all codes are linear (over $ \F $) unless otherwise stated.

\begin{definition} [\textbf{Multiplicity codes} \cite{koppartyMultiplicityDef}]
Take $ n = |S| $ and $ t = \binom{m+r-1}{m} $. We define the multiplicity code on $ S $ of multiplicity $ r $ and degree $ k $ as
$$ \mathcal{M}(S,r,k) = \left\lbrace \left( \left( f^{(\ii)}(\aaa) \right)_{|\ii| < r} \right)_{\aaa \in S} : f \in \F[\xx]_{<k} \right\rbrace \subseteq (\F^t)^n . $$
\end{definition}
For easy notation, we will sometimes omit the inner parentheses. The parameter $r$ bounds the derivatives, and $k$ bounds the degrees of the evaluating polynomials. When the number of variables is $m=1$ and $k < | S_1 |$, then $k$ is the dimension of the multiplicity code. But in general, $k$ does not represent the dimension.
\begin{example}
Assume $\F_q = \{a_1, \ldots, a_q \}$. The multiplicity code $\mathcal{M}(\F_q,1,k)$ is the classical Reed-Solomon code
\[\mathcal{M}(\F_q,1,k) = \{\left(f\left(a_1\right), \ldots, f\left(a_q\right) \right) : f \in \F_q[x]_{<k} \} \subseteq \F_q^q . \]
We can also see that
\[\mathcal{M}(\F_q,2,k) = \left\{\left((f\left(a_1\right),f^{(1)}\left(a_1\right)), \ldots, (f\left(a_q\right),f^{(1)}\left(a_q\right)) \right) : f \in \F_q[x]_{<k} \right\} \subseteq (\F_q^2)^q . \]
\end{example}

A method to lower bound the minimum distance of multiplicity codes over grids is the Schwartz-Zippel bound with multiplicities given in \cite[Lemma 8]{extensions}.

\begin{definition} [\textbf{Multiplicity}] \label{def multi codes}
We define the multiplicity of $ \aaa \in \F^m $ in $ f \in \F[\xx] $, denoted $ m(f,\aaa) $, as the largest $ r \in \mathbb{N} $ such that $ f^{(\ii)}(\aaa) = 0 $, for all $ \ii \in \mathbb{N}^m $ such that $ |\ii| < r $. 
\end{definition}

\begin{lemma} [\textbf{Schwartz-Zippel} \cite{extensions}] \label{lemma SZ bound}
If $ s = |S_1| = \cdots = |S_m| $, then for any $ f \in \F[\xx] $,
$$ \sum_{\aaa \in S} m(f,\aaa) \leq \deg(f) s^{m-1}. $$
\end{lemma}

As a consequence, we obtain the following result, which is essentially \cite[Lemma 3.5]{koppartyMultiplicityDef}.

\begin{proposition}[\hspace{1pt}\cite{koppartyMultiplicityDef}] \label{prop lower bound dist mult SZ}
If $ s = |S_1| = \ldots = |S_m| $ and $ n = s^m = |S| $, then
$$ \dd(\mathcal{M}(S,r,k)) \geq s^m - \frac{(k-1)s^{m-1}}{r} = \left( 1 - \frac{k-1}{rs} \right) n . $$ 
\end{proposition}

%
%

\subsection{General multiplicity codes}
In order to compute duals of multiplicity codes, we will consider general multiplicity codes for decreasing or down sets of Hasse derivatives, as in \cite{geil2019}.

\begin{definition}[\hspace{1pt}\cite{geil2019}]
The set $ \J \subseteq \mathbb{N}^m $ is decreasing if whenever $ \ii \in \J $ and $ \jj \in \mathbb{N}^m $ are such that $ \jj \leq \ii $, it holds that $ \jj \in \J $, where $ \leq $ is the coordinatewise partial order in $ \mathbb{N}^m $. We will typically consider two families of decreasing sets: The classical multiplicity set
$ \J_r = \{ \ii \in \mathbb{N}^m : |\ii| < r \}$
for $ r \in \mathbb{N} $, and the coordinatewise or box multiplicity set
$\J_\rr = \{ \ii \in \mathbb{N}^m : \ii \leq \rr \}$
for $ \rr \in \mathbb{N}^m$. Figure~\ref{25.04.24} shows an example of a multiplicity set and a box multiplicity set.
\end{definition}

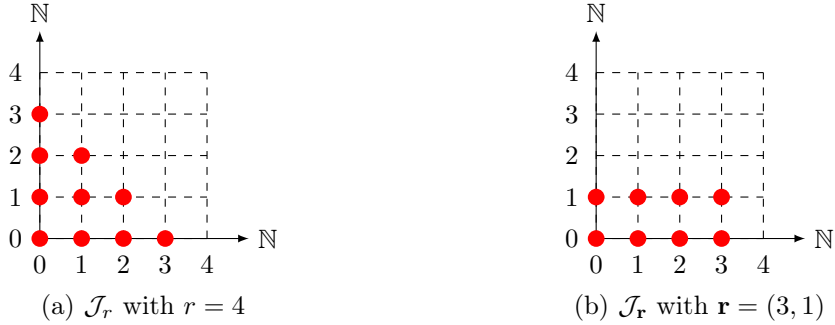
\begin{figure}[h!]
\begin{minipage}[t]{0.49\textwidth}
\begin{center}
\begin{tikzpicture}[scale=0.55]
\draw [-latex] (0,0) -- (5,0)node[right] {$\mathbb{N}$};
\draw [-latex] (0,0) -- (0,5)node[above] {$\mathbb{N}$};
\foreach \x in {0,...,4}{
\draw [dashed] (0,\x) -- (4,\x);
\draw [dashed] (\x,0) -- (\x,4);
\fill [](\x,-0.2)node[below]{$\x$};
\fill [](-0.2,\x)node[left]{$\x$};}

\foreach \x in {0,...,3}{
\foreach \y in {0}{
\fill [color=red](\x,\y) {circle(.2cm)};}}

\foreach \x in {0,...,2}{
\foreach \y in {1}{
\fill [color=red](\x,\y) {circle(.2cm)};}}

\foreach \x in {0,1}{
\foreach \y in {2}{
\fill [color=red](\x,\y) {circle(.2cm)};}}

\foreach \x in {0}{
\foreach \y in {3}{
\fill [color=red](\x,\y) {circle(.2cm)};}}
\end{tikzpicture}
\vskip 0cm
(a) $\J_r$ with $r=4$
\end{center}
\end{minipage}
\begin{minipage}[t]{0.49\textwidth}
\begin{center}
\begin{tikzpicture}[scale=0.55]
\draw [-latex] (0,0) -- (5,0)node[right] {$\mathbb{N}$};
\draw [-latex] (0,0) -- (0,5)node[above] {$\mathbb{N}$};
\foreach \x in {0,...,4}{
\draw [dashed] (0,\x) -- (4,\x);
\draw [dashed] (\x,0) -- (\x,4);
\fill [](\x,-0.2)node[below]{$\x$};
\fill [](-0.2,\x)node[left]{$\x$};}

\foreach \x in {0,...,3}{
\foreach \y in {0,1}{
\fill [color=red](\x,\y) {circle(.2cm)};}}
\end{tikzpicture}
\vskip 0cm
(b) $\J_\rr$ with $\rr = (3,1)$
\end{center}
\end{minipage}
\caption{
(a) The multiplicity set $\J_r$ with $r=4$.
(b) The box multiplicity set $\J_\rr$ with $\rr = (3,1)$.}
\label{25.04.24}
\vskip 0. cm
\end{figure}

We now define general multiplicity codes, an extension of Definition \ref{def multi codes} that has already been considered in \cite[Def. 4.16]{geil2019}.

\begin{definition}[\textbf{General multiplicity codes} \cite{geil2019}] \label{def multi codes general}
Given a decreasing set $ \J \subseteq \mathbb{N}^m $, we define the multiplicity code on $ S $ with multiplicity set $ \J $ and degree $ k $ as
$$ \mathcal{M}(S,\J,k) = \left\lbrace \left( \left( f^{(\ii)}(\aaa) \right)_{\ii \in \J} \right)_{\aaa \in S} : f \in \F[\xx]_{<k} \right\rbrace \subseteq \left( \F^{|\J|} \right)^{|S|} . $$
\end{definition}
For easy notation, we will sometimes omit the inner parentheses. Clearly, $ \mathcal{M}(S,r,k) = \mathcal{M}(S,\J_r,k) $, hence Definition \ref{def multi codes general} generalizes Definition \ref{def multi codes}. Furthermore, if $ \rr = (r,\ldots, r) $ and $ \ones = (1, \ldots, 1) $, then $ \mathcal{M}(S,r,k) $ is a punctured code of $ \mathcal{M}(S,\J_{\rr- \ones},k) $. However, the dual of the latter is easier to compute in the multivariate case (note they coincide in one variable). This will be our approach in Subsections \ref{subsec duals m=r=2} and \ref{subsec duals general}. In the following, it will be very useful to remember at every moment that
$$ \J_{\rr - \ones} = [0,r-1]^m. $$
We will write $ \J_{\rr - \ones} $ or $ [0,r-1]^m $ interchangeably in the remainder of the manuscript.

\subsection{Isometries for the folded metric}

The linear isometries for the folded Hamming metric are known since the folded Hamming metric is a particular case of the sum-rank metric, see \cite[Th. 2]{sr-hamming}. 

\begin{proposition} [\hspace{1pt}\cite{sr-hamming}] \label{prop isometries}
Let $ \phi : \F^{tn} \longrightarrow \F^{tn} $ be a linear vector space isomorphism. Then $ \phi $ is an isometry for the $ t $-folded Hamming metric (i.e., $ \w(\phi(\cc)) = \w(\cc) $ for all $ \cc \in \F^{tn} $) if, and only if, there exist invertible matrices $ A_1, \ldots, A_n \in {\rm GL}_t(\F) $ and a permutation $ \sigma : [n] \longrightarrow [n] $ such that, for all $ \cc_1, \ldots, \cc_n \in \F^t $,
$$ \phi(\cc_1, \ldots, \cc_n) = \left( \cc_{\sigma(1)} A_1, \ldots, \cc_{\sigma(n)}A_n \right). $$
\end{proposition}

\begin{definition}
The codes $ \C , \C^\prime \subseteq \F^{tn} $ are equivalent if there exists a linear isometry $ \phi : \F^{tn} \longrightarrow \F^{tn} $ for the folded metric such that $ \C^\prime = \phi (\C) $. We also say that $ \phi $ is an equivalence.
\end{definition}

In particular, two codes are equivalent if, and only if, so are their duals. More concretely, we have the following. The proof is straightforward.

\begin{proposition} \label{prop duals of equivalent}
Given invertible matrices $ A_1, \ldots, A_n \in {\rm GL}_t(\F) $ and a permutation $ \sigma : [n] \longrightarrow [n] $, define the equivalences $ \phi, \psi : \F^{tn} \longrightarrow \F^{tn} $ by
\begin{equation*}
\begin{split}
\phi(\cc_1, \ldots, \cc_n) & = \left( \cc_{\sigma(1)} A_1, \ldots, \cc_{\sigma(n)}A_n \right), \\
\psi(\cc_1, \ldots, \cc_n) & = \left( \cc_{\sigma(1)} (A_1^\intercal)^{-1}, \ldots, \cc_{\sigma(n)} (A_n^\intercal)^{-1} \right),
\end{split}
\end{equation*}
for $ \cc_1, \ldots, \cc_n \in \F^t $. Then for any linear code $ \C \subseteq \F^{tn} $, 
$$ \phi(\C)^\perp = \psi (\C^\perp) . $$
\end{proposition}

\begin{remark} \label{remark generalized codes}
Any family $ \Gamma $ of codes $ \C \subseteq \F^{tn} $ can be generalized to include all equivalent codes $ \overline{\Gamma} = \{ \phi(\C) : \C \in \Gamma, \phi \textrm{ is an equivalence} \} $, as is done with generalized Reed--Solomon codes. However, by Proposition \ref{prop duals of equivalent}, if we are able to compute the family $ \Gamma^\perp $ of codes that are dual to some code in $ \Gamma $, then we have also computed the family $ \overline{\Gamma}^\perp $ of codes that are dual to some code in $ \overline{\Gamma} $. In Subsection \ref{subsec duals general}, we compute the duals of multiplicity codes as in Definition \ref{def multi codes}, and by this remark, we have computed the duals of ``generalized multiplicity codes'', i.e., codes equivalent to a multiplicity code. 
\end{remark}

\section{Hermite interpolation} \label{sec hermite}
In this section, we study Hermitian interpolator polynomials, whose coefficients will be used to compute the duals. We will see that a Hermitian interpolation basis for the set $ \J_{\rr-\ones} $ is formed by polynomials with separated variables. Such a Hermitian interpolation basis is easier than for the set $ \J_{r} $.

Let $ T = \{ a_1, \ldots, a_s \} $ be a set of (distinct) elements of $ \F $.

\subsection{Lagrange interpolation}\label{25.04.30}
Define the map
$ev\colon \F[x] \to \F^{s},
f \mapsto ev(f) = \left( f(a_1),\ldots,f(a_s)\right).$
For any $s$ elements $b_1, \ldots, b_s$ in $\F$, the classic interpolation asks for a polynomial $f \in \F[x]$ such that $ \deg(f) < s $ and
\begin{equation}\label{25.04.25}
ev(f) = (b_1,\ldots,b_s).
\end{equation}
Lagrange interpolation easily solves this problem. Indeed, let 
$G(x) = \prod_{a\in T}(x-a)$ be the vanishing polynomial of $T$. For $i \in [s]$, the polynomial $\displaystyle h_i(x) = \frac{G(x)}{G^{(1)}(a_i) (x-a_i)}$ satisfies $ \deg(h_i) = s-1 $ and $\displaystyle h_i(a_j)=\delta_{i,j}$, the Kronecker delta function. Thus, \[f(x) = \sum_{i=1}^s h_i(x)b_i =
\left(h_1(x),\ldots,h_s(x)\right) \cdot
\left(b_1,\ldots,b_s\right)\]
satisfies Eq.~(\ref{25.04.25}) and $ \deg(f) < s $. We will call the polynomials $h_1, \ldots, h_s \in \F[x]$ the {\bf indicator functions} of $T$.

\subsection{Single-variable Hermite interpolation}
{\bf Order 1.} Let $b_1, b_1^{\prime}, \ldots, b_s, b_s^{\prime}$ be any $2s$ elements in $\F$. Define the map
\begin{align*}
ev^{(1)}\colon \F[x] &\to \F^{2s}\\
f &\mapsto 
ev^{(1)}(f) = \left( f(a_1),f^{(1)}(a_1),\ldots,f(a_s),f^{(1)}(a_s) \right).
\end{align*}
For easy notation, we are omitting the inner parentheses. The goal of the Hermite interpolation of order 1 is to find a polynomial $f \in \F[x]$ such that $ \deg(f) \leq 2s -1 $ and
\begin{equation}\label{25.04.26}
ev^{(1)}(f) = (b_1, b_1^{\prime}, \ldots,b_s, b_s^{\prime}).
\end{equation}
The classic interpolation presented in subsection~\ref{25.04.30} relies on $h_1, \ldots, h_s$. To find $f$ that satisfies Eq.~(\ref{25.04.26}), we now look for $2s$ polynomials $h_{1,0},h_{1,1},\ldots,
h_{s,0},h_{s,1}$ of degree $ \leq 2s-1 $ whose images under the map $ev^{(1)}$ are the standard vectors.
\begin{align*}
ev^{(1)}(h_{1,0})&=\left(1,0,\ldots,0,0\right)\\
ev^{(1)}(h_{1,1})&=\left(0,1,\ldots,0,0\right)\\
\vdots\\
ev^{(1)}(h_{s,0})&=\left(0,0,\ldots,1,0\right)\\
ev^{(1)}(h_{s,1})&=\left(0,0,\ldots,0,1\right).
\end{align*}
In other words, $h_{i,0}(a_j)$ is the Kronecker delta, but its derivative is $0$ on every element of $T$. The polynomial $h_{i,1}(a_j)$ is $0$ on every element of $T$, but its derivative is the Kronecker delta.

The following results shows how to find $h_{i,1}(x)$ and $h_{i,0}(x)$.
\begin{proposition} [\textbf{Single-variable Hermite interpolation of order 1}]\label{25.05.01}
Let $\displaystyle G(x)$ be the vanishing polynomial and $h_1, \ldots, h_s$ the indicator functions of $T$. There exists a unique polynomial $f \in \F[x]_{\leq 2s-1}$ such that
\[ev^{(1)}(f)=(b_1, b_1^{\prime}, \ldots, b_s, b_s^{\prime}).\]
Moreover,
$f = 
\left(h_{1,0},h_{1,1},\ldots,
h_{s,0},h_{s,1}\right) \cdot
\left(b_1, b_1^{\prime}, \ldots, b_s, b_s^{\prime}\right)$,
where 
\[
\begin{pmatrix}
h_{i,0}(x)\\
h_{i,1}(x)
\end{pmatrix}
=
\begin{pmatrix}
1 & -h_{i}^{2(1)}(a_i)\\
0 & 1
\end{pmatrix}
\begin{pmatrix}
1\\
x-a_i
\end{pmatrix}
h_i^2(x),\]
for $i \in [s]$, and $h_{i}^{2(1)}(a_i)$ is the evaluation of the derivative of $h_{i}^2(x)$ at $x=a_i$.
\end{proposition}
\begin{proof}
By definition, we can see that
\begin{align}
h_{i,1}(x) &= (x-a_i)h_i^2(x) \qquad \text{ and }\\
h_{i,0}(x) &= h_{i}^2(x) - h_{i}^{2{(1)}}(a_i)h_{i,1}(x).
\end{align}
We have that $h_{i,0}(a_j)=\delta_{i,j}$ and $h_{i,1}(a_j)=0$ for $j\in[s]$. By taking derivatives and using Proposition~\ref{prop leibniz}, we obtain
\begin{align*}
h_{i,1}^{(1)}(x)
&= (x-a_i)^{(1)} h_{i}^2(x) + (x-a_i)h_{i}^{2(1)}(x)\\
&= h_{i}^2(x) + (x-a_i) 2h_{i}(x)h_i^{(1)}(x)\\
h_{i,0}^{(1)}(x) &= h_{i}^{2(1)}(x) - h_{i}^{2{(1)}}(a_i)h_{i,1}^{(1)}(x).
\end{align*}
It is straightforward to see that $h_{i,1}^{(1)}(a_j)=\delta_{i,j}$ and $h_{i,0}^{(1)}(a_j)=0$ for $j\in[s]$. Note the polynomial
\[f = 
\left(h_{1,0},h_{1,1},\ldots,
h_{s,0},h_{s,1}\right) \cdot
\left(b_1, b_1^{\prime}, \ldots, b_s, b_s^{\prime}\right)\]
satisfies Eq.~(\ref{25.04.26}) and has degree $ \leq 2s-1 $. The uniqueness of $f\in\F[x]_{\leq 2s-1}$ is due to the fact that $f$ is the solution of a system of $2s$ equations and $2s$ indeterminates.
\end{proof}

\begin{remark}\label{25.05.08}
We have
\[ \begin{pmatrix}
h_{i,0}(x)\\
h_{i,1}(x)
\end{pmatrix} =
\begin{pmatrix}
1 & -2\frac{G^{(2)}(a_i)}{G^{(1)}(a_i)}\\
0 & 1
\end{pmatrix} \begin{pmatrix}
1\\
x-a_i
\end{pmatrix} h_i^2(x).\]
Indeed, by Proposition~\ref{25.05.01}, we only need to compute $h_{i}^{2(1)}(a_i)$. Note $\displaystyle (x-a_i)h_i(x)=\frac{G(x)}{G^{(1)}(a_i)}$. By Proposition~\ref{prop leibniz}, we have
\begin{align*}
(x-a_i)h_i^{(2)}(x) + h_i^{(1)}(x)
 = \frac{G^{(2)}(x)}{G^{(1)}(a_i)}
\qquad \text{and} \qquad h_{i}^{2(1)}(x)  = 2 h_{i}(x)h_{i}^{(1)}(x).
\end{align*}
Thus,
\begin{align*}
h_{i}^{2(1)}(a_i) = 2 h_{i}(a_i)h_{i}^{(1)}(a_i)
= 2 h_{i}^{(1)}(a_i)
= 2\frac{G^{(2)}(a_i)}{G^{(1)}(a_i)}.
\end{align*}
\end{remark}

\noindent {\bf Order 2}. Let $b_1, b_1^{\prime}, b_1^{\prime\prime}, \ldots, b_s, b_s^{\prime}, b_s^{\prime\prime}$ be any $3s$ elements in $\F$. Define the map
\begin{align*}
ev^{(2)}\colon \F[x] &\to \F^{3s}\\
f &\mapsto 
ev^{(2)}(f) = \left( f(a_1),f^{(1)}(a_1),f^{(2)}(a_1),\ldots,f(a_s),f^{(1)}(a_s),f^{(2)}(a_s) \right).
\end{align*}
\begin{proposition} [\textbf{Single-variable Hermite interpolation of order 2}] \label{prop hermite order 2}
Let $\displaystyle G(x)$ be the vanishing polynomial and $h_1, \ldots, h_s$ the indicator functions of $ T $. There exists a unique polynomial $f \in \F[x]_{\leq 3s-1}$ such that
\[ev^{(2)}(f)=(b_1, b_1^{\prime}, b_1^{\prime\prime}, \ldots, b_s, b_s^{\prime}, b_s^{\prime\prime}).\]
Moreover,
$f = 
\left(h_{1,0},h_{1,1},h_{1,2},\ldots,
h_{s,0},h_{s,1},h_{s,2}\right) \cdot
\left(b_1, b_1^{\prime}, b_1^{\prime\prime}, \ldots, b_s, b_s^{\prime}, b_s^{\prime\prime} \right)$,
where 
\[
\begin{pmatrix}
h_{i,0}(x)\\
h_{i,1}(x)\\
h_{i,2}(x)
\end{pmatrix}
=
\begin{pmatrix}
1 & -h_{i}^{3(1)}(a_i) & -h_{i}^{3(2)}(a_i) + h_i^{3(1)}(a_i)^2 \\
0 & 1 & -h_{i}^{3(1)}(a_i)\\
0 & 0 & 1\\
\end{pmatrix}
\begin{pmatrix}
1\\
x-a_i\\
\left(x-a_i\right)^2
\end{pmatrix}
h_i^3(x),\]
for $i\in[s]$, and $h_{i}^{3(j)}(a_i)$ is the value of the $j$th derivative of $h_{i}^3(x)$ at $x=a_i$.
\end{proposition}
\begin{proof}
It is easy to see that the polynomial $h_{i,2}(x) = (x-a_i)^2h^3_i(x)$ and its derivative vanish on every element of $T$. By Proposition~\ref{prop leibniz},
\[h_{i,2}^{(2)}(x) = (x-a_i)^2h_i^{3(2)}(x) + 2(x-a_i)h_i^{3(1)}(x) + h_i^3(x).\]
Thus, $h_{i,2}^{(2)}(a_j) = \delta_{i,j}$. As 
$\left((x-a_i)h_i^3\right)^{(2)}(x) = 
(x-a_i)h_i^{3(2)}(x) + h_i^{3(1)}(x)$
by Proposition~\ref{prop leibniz},
we obtain
\[\left((x-a_i)h_i^3\right)^{(2)}(a_i) = 
(a_i-a_i)h_i^{3(2)}(a_i) + h_i^{3(1)}(a_i) = h_i^{3(1)}(a_i).\]
The last equation implies that
\begin{align*}
h_{i,1}(x) &= (x-a_i)h_i^3(x) - \left((x-a_i)h_i^3\right)^{(2)}(a_i)h_{i,2}(x).
\end{align*}
By the properties of $h_i(x)$ and $h_{i,2}(x)$, we have that $h_{i,1}(x)$ and $h_{i,1}^{(2)}(x)$ vanish on every element of $T$. In addition, $h_{i,1}^{(1)}(a_j)=\delta_{i,j}$. From the definitions, we can see that
\begin{align*}
h_{i,0}(x) &= h_{i}^3(x) - h_{i}^{3{(1)}}(a_i)h_{i,1}(x) - h_{i}^{3{(2)}}(a_i)h_{i,2}(x).
\end{align*}
By the properties of $h_i(x)$, $h_{i,1}(x)$, and $h_{i,2}(x)$, we obtain the result. Note the polynomial
\[f = 
\left(h_{1,0},h_{1,1},h_{1,2},\ldots,
h_{s,0},h_{s,1},h_{s,2}\right) \cdot
\left(b_1, b_1^{\prime}, b_1^{\prime\prime}, \ldots, b_s, b_s^{\prime}, b_s^{\prime\prime} \right)\]
satisfies $ev^{(2)}(f)=(b_1, b_1^{\prime}, b_1^{\prime\prime}, \ldots, b_s, b_s^{\prime}, b_s^{\prime\prime})$ and $\deg(f) \leq 3s-1 $. The uniqueness of $f\in \F[x]_{\leq 3s-1}$ is due to the fact that $f$ is the solution of a system of $3s$ equations and $3s$ indeterminates.
\end{proof}

\noindent {\bf Order $r$ (general case)}. Let $b_{10}, b_{11},\ldots, b_{1r}, \ldots, b_{s0}, b_{s1},\ldots, b_{sr}$ be any $s(r+1)$ elements in $\F$. Define the map
\begin{align*}
ev^{(r)}\colon \F[x] &\to \F^{(r+1)s}\\
f &\mapsto 
ev^{(r)}(f) = \left( f(a_1),f^{(1)}(a_1),\ldots,f^{(r)}(a_1),\ldots,f(a_s),f^{(1)}(a_s),\ldots,f^{(r)}(a_s) \right).
\end{align*}
\begin{theorem} [\textbf{Single-variable Hermite interpolation}] \label{prop hermite interpolation univariate}
Let $\displaystyle G(x)$ be the vanishing polynomial and $h_1, \ldots, h_s$ the indicator functions of $ T $. There exists a unique polynomial $f \in \F[x]_{\leq (r+1)s-1}$ such that
\[ev^{(r)}(f) = \left(b_{10}, b_{11},\ldots, b_{1r}, \ldots,b_{s0}, b_{s1},\ldots, b_{sr}\right).\]
Moreover, $f = 
\left(h_{1,0},\ldots,h_{1,r},\ldots,
h_{s,0},\ldots,h_{s,r}\right) \cdot
\left(b_{10}, b_{11},\ldots, b_{1r}, \ldots,b_{s0}, b_{s1},\ldots, b_{sr}\right)$, where we define, recursively in $ n = r,r-1, \ldots, 1,0 $, in decreasing order, $ h_{i,r}(x) = (x-a_i)^r h_i^{r+1}(x) $ and
$$ h_{i,n}(x) = (x-a_i)^n h_i^{r+1}(x) - \sum_{k=n+1}^r h_i^{r+1(k-n)}(a_i) h_{i,k}(x), $$
for $i\in[s]$. Here, $h_{i}^{r+1(j)}(a_i)$ is the value of the $j$th derivative of $h_{i}^{r+1}(x)$ at $x=a_i$.
\end{theorem}
\begin{proof}
Let $\left(e_{1,0},\ldots,e_{1,r},\ldots,e_{s,0},\ldots,e_{s,r}\right)^\intercal$ be the $(r+1)s \times (r+1)s$ identity matrix. In other words, the $e_{i,n}$'s represent the standard vectors for $i \in [1,s]$ and $n\in[0,r]$. We will show that $ev^{(r)}(h_{i,n})=e_{i,n}$, for all $ i \in [s] $, by induction in $n=r,\ldots,1,0$ (in decreasing order).

(Case $n=r$) The polynomial $h_{i,r}(x) = (x-a_i)^rh^{r+1}_i(x)$ and its first $r-1$ derivatives vanish on every element of $T$. By Proposition~\ref{prop leibniz},
\[h_{i,r}^{(r)}(x) = (x-a_i)h_i(x)g(x) + h_i^{r+1}(x),\]
for some $g(x) \in \F[x]$. Thus, $h_{i,r}^{(r)}(a_j) = h_i^{r+1}(a_j) = \delta_{i,j}$, meaning $ev^{(r)}(h_{i,r})=e_{i,r}$ for $i \in [s]$.

(Case $n \leq r-1$) Fix $ i \in [s] $. We assume by induction in $ n $ that $ev^{(r)}(h_{i,t})=e_{i,t}$, for all $ t \in [n+1,r] $. Given $ h_{i,n}(x) $ as in the proposition, we now show that $ev^{(r)}(h_{i,n})=e_{i,n}$. We consider different cases in order to compute $ h_{i,n}^{(\ell)}(a_j) $.

If $ 0 \leq \ell \leq n-1 $, then we have $ ((x-a_i)^nh_i^{r+1}(x))^{(\ell)}(a_j) = 0 $ and $ h_{i,n+1}^{(\ell)}(a_j) = \ldots = h_{i,r}^{(\ell)}(a_j) = 0 $, thus $ h_{i,n}^{(\ell)}(a_j) = 0 $, for all $ j \in [s] $.

If $ \ell = n $, then $ ((x-a_i)^nh_i^{r+1}(x))^{(n)}(a_j) = h_i^{r+1}(a_j) $ by Proposition \ref{prop leibniz}, and $ h_{i,n+1}^{(n)}(a_j) = \ldots = h_{i,r}^{(n)}(a_j) = 0 $, thus $ h_{i,n}^{(n)}(a_j) = h_i^{r+1}(a_j) = \delta_{i,j} $, for all $ j \in [s] $.

If $ n+1 \leq \ell \leq r $, then $ ((x-a_i)^nh_i^{r+1}(x))^{(\ell)}(a_j) = h_i^{r+1(\ell-n)}(a_j) $ by Proposition \ref{prop leibniz}, which is $ 0 $ if $ j \neq i $. Now, by the induction hypothesis, we have
\begin{equation*}
\begin{split}
h_{i,n}^{(\ell)}(a_j) & = ((x-a_i)^nh_i^{r+1}(x))^{(\ell)}(a_j) - \sum_{k=n+1}^r h_i^{r+1(k-n)}(a_i) h_{i,k}^{(\ell)}(a_j) \\
 & = h_i^{r+1(\ell-n)}(a_j) - h_i^{r+1(\ell-n)}(a_j) \delta_{i,j} \\
 & = 0,
\end{split}
\end{equation*}
for all $ j \in [s] $, and we are done.

\end{proof}

\begin{remark}
Notice that, since $ h_{i,r}(x) $ only depends on $ h_i(x) $, then by induction all of the polynomials $ h_{i,n}(x) $, for $ n \in [0,r] $, can be expressed as a function of only $ h_i(x) $ (and its derivatives at $ x=a_i $), as in Propositions \ref{25.05.01} and \ref{prop hermite order 2}. However, we omit this formula in general for brevity.
\end{remark} 

\begin{remark}\label{25.05.04}
In case the set $T$ and the element $a_i$ are relevant, we also denote $h_{i,j}$ from Theorem~\ref{prop hermite interpolation univariate} by $h_{T,a_i,j}$. Note that the polynomials $h_{i,j}$'s are playing the role of indicator functions with derivatives.
\end{remark}
\begin{remark}\label{25.05.05}
From Theorem~\ref{prop hermite interpolation univariate}, we have that
\[h_{i,r}(x) = \left(x-a_i\right)^rh_{i}^{r+1}(x)
= (x-a_i)^{r} \frac{G(x)^{r+1}}{\left(G^{(1)}(a_i)\right)^{r+1}\left(x-a_i\right)^{r+1}}.\]
In particular, $\coef(h_{i,r},x^{(r+1)s-1}) = \left(G^{(1)}(a_i)\right)^{-r-1} \neq 0 $.
\end{remark}

\begin{remark} \label{remark coef hermite r=2}
In the case $ T = \F_q $, we have $G(x)=x^q-x$. Hence $G^{(1)}(x)=-1$. Therefore, 
$$ \coef(h_{i,r},x^{(r+1)q-1}) = (-1)^{r+1}.$$
Moreover, by Remark \ref{25.05.08}, for $ r=1 $, we have that $ h_{i,0}(x) = h_i^2(x) $, since $ 2G^{(2)}(a_i) = 0 $, for all $ i \in [q] $ in that case. In particular, $ \deg(h_{i,0}) = 2q-2 $, thus for $ r=1 $,
$$ \coef(h_{i,0},x^{2q-1}) = 0. $$ 

In the case $ T = \F^*_q $, we have $G(x)=x^{q-1}-1$. Hence $G^{(1)}(x)=(q-1)x^{q-2}=-x^{q-2}=-x^{-1}$ and $G^{(1)}(a_i)=(-a_i)^{-1}$. Therefore, 
\[
\coef(h_{i,r},x^{(r+1)(q-1)-1}) = \left(-a_i\right)^{r+1}.
\]


\end{remark}

\subsection{Multivariate Hermite interpolation}
Recall that $ S = S_1 \times \cdots \times S_m $ represents a Cartesian set, where every $ S_i \subseteq \F $ is finite. For a positive integer $ r $, $\J_{\rr - \ones} = [0,r-1]^m$. Let $\{b_{\aaa,\ii}\}_{\aaa \in S, \ii \in \J_{\rr - \ones}}$ be a multiset of $|S|r^m$ elements in $\F$. Similarly to the single-variable case, we define the evaluation map
\begin{align*}
ev^{(\J_{\rr - \ones})}\colon \F[\xx] &\to \F^{|S|r^m}\\
f &\mapsto 
ev^{(\J_{\rr - \ones})}(h) = \left( \left( h^{(\ii)}(\aaa) \right)_{\ii \in \J_{\rr - \ones}} \right)_{\aaa \in S}.
\end{align*}

\begin{corollary} [\textbf{Rectangular Hermite interpolation}] \label{cor hermite interpolation multivariate}
Assume $|S_i| = s_i$.
There exists a unique $ h \in \F[\xx] $ with $ \deg_{x_j} (h) \leq rs_j-1 $ for $ j \in [m] $ such that
\[ev^{(\J_{\rr - \ones})}(h)
= \left(\left(b_{\aaa,\ii}\right)_{\ii \in \J_{\rr - \ones}}\right)_{\aaa \in S}.\]
Moreover, $h=
\left(h_{\aaa,\ii}\right)_{\aaa \in S, \hspace{0.1 cm} \ii \in \J_{\rr - \ones}}
\cdot \left(b_{\aaa,\ii}\right)_{\aaa \in S, \hspace{0.1 cm} \ii \in \J_{\rr - \ones}}$, where
$$ h_{\aaa,\ii}(\xx) = \prod_{j=1}^m h_{S_j,a_j,i_j}(x_j) \in \F[\xx],$$
and
$h_{S_j,a_j,i_j}$ is defined in Remark~\ref{25.05.04} (and Theorem \ref{prop hermite interpolation univariate}) for $\aaa = (a_1, \ldots, a_m) \in S$ and $\ii = (i_1, \ldots, i_m) \in \J_{\rr-\ones}$.
\end{corollary}
\begin{proof}
The map that sends a polynomial $ h \in \F[\xx] $ with $ \deg_{x_j}(h) \leq rs_j-1 $, for $ j \in [m] $, to the evaluations $ ((h^{(\ii)}(\aaa))_{\ii \in \J_{\rr-\ones}})_{\aaa \in S} $ is a vector space isomorphism from the space of such polynomials to $ (\F^{|\J_{\rr-\ones}|})^{|S|} $. The map is surjective by using Theorem~\ref{prop hermite interpolation univariate}. Thus, it is an isomorphism since the dimensions of both spaces are the same over $\F$.
\end{proof}
\begin{remark}\label{25.05.09}
We also denote $h_{\aaa,\ii}$ by $h_{S,\aaa,\ii}$ in case the set $S$ is relevant.
\end{remark}
\begin{remark}
Notice that
$$ \coef(h_{\aaa, \rr-\ones}, x_1^{rs_1-1}\cdots x_m^{rs_m-1}) = \prod_{j=1}^m \coef(h_{S_j,a_j,i_j}(x_j), x_j^{rs_j-1}). $$
Therefore, by Remark~\ref{25.05.05}, we have that
\begin{equation}
\coef(h_{\aaa, \rr - \ones},x_1^{rs_1-1}\cdots x_m^{rs_m-1}) = \left(G_1^{(1)}(a_1)\cdots G_m^{(1)}(a_m)\right)^{-r} \neq 0,
\label{eq coef of hermite basis multivariate}
\end{equation}
for $ \aaa = (a_1,\ldots,a_m) \in S $, where $G_i(x)$ the vanishing polynomial of $S_i$.
\end{remark}

\section{Dimension} \label{sec dimension}
In this section, we compute the dimensions of general multiplicity codes (Definition \ref{def multi codes general}) using Gr{\"o}bner basis tools. The dimension plays an important role in computing the dual.

\begin{definition}
For any finite decreasing set $\mathcal{J} \subseteq \mathbb{N}^m$, we define
$$
\mathcal{B}_\mathcal{J}=\{ \ii \in \mathbb{N}^m \setminus \J : \jj  \in \mathbb{N}^m \setminus \J \text{ and } \jj \leq \ii \implies \ii=\jj\}.
$$
\end{definition}

We next consider the polynomials whose derivatives of orders in $ \J $ vanish in $ S $. The following definition is \cite[Def. 3.18]{geil2019}.

\begin{definition} [\hspace{1pt}\cite{geil2019}] \label{def ideal of a set with multi}
Given a finite decreasing set $\mathcal{J} \subseteq \mathbb{N}^m$, we define
$$ I(S; \J) = \left\lbrace f \in \F[\xx] : f^{(\ii)}(\aaa) = 0, \textrm{ for all } \aaa \in S \textrm{ and all } \ii \in \J \right\rbrace . $$
\end{definition}

Using that $ \J $ is decreasing and the Leibniz rule (Proposition \ref{prop leibniz}), it is easy to see that $ I(S;\J) $ is an ideal in $ \F[\xx] $. In \cite[Th. 4.7]{geil2019}, a Gr{\"o}bner basis for such an ideal is computed for any monomial ordering.

\begin{theorem}[\hspace{1pt}\cite{geil2019}] \label{T:grobnerbasis}
Take $G_j(x_j)=\prod_{\alpha\in S_j}(x_j-\alpha)$, for $ j \in [m] $. For any finite decreasing set $\mathcal{J}\subseteq \mathbb{N}^m$, the family
$$
\mathcal{F}=\left \{ \prod_{j=1}^mG_j(x_j)^{u_j} : (u_1,u_2,\dots,u_m)\in \mathcal{B}_\mathcal{J}  \right \}
$$
is a reduced Gröbner basis for the ideal $I(S;\mathcal{J})$ with respect to any monomial ordering (see \cite[Def. 2.5.5 \& Def. 2.7.4]{clo1} for the definitions). 
\end{theorem}

We will also need the concept of footprint \cite{footprints}, which is well known in the commutative algebra literature.

\begin{definition} [\hspace{1pt}\cite{footprints}] \label{def footprint}
Given a monomial ordering $ \prec $ and an ideal $ I \subseteq \F[\xx] $, we define its footprint with respect to $ \prec $ as 
$$ \Delta_\prec(I) = \{ \xx^\ii : \xx^\ii \notin \langle \inn_\prec(I) \rangle \}, $$
where $ \inn_\prec(I) = \{ \inn_\prec(f) : f \in I \} $, i.e., $ \langle \inn_\prec(I) \rangle $ is the initial ideal of $ I $ with respect to $ \prec $.
\end{definition}

\begin{remark} \label{remark footprint of box}
For $ \rr = r \ones $ and any monomial ordering $ \prec $, using the Gr{\"o}bner basis from Theorem \ref{T:grobnerbasis}, it is straightforward to verify that
$$ \Delta_\prec(I(S;\J_{\rr-\ones})) = \left\lbrace \xx^\ii : \ii \in \prod_{j=1}^m [0,rs_j-1] \right\rbrace. $$
In particular, any $ f \in \langle \Delta_\prec(I(S;\J_{\rr-\ones})) \rangle_\F $ satisfies $ \deg_{x_j}(f) \leq rs_j - 1 $, for all $ j \in [m] $.
\end{remark}
%

Footprints are useful for reducing polynomials while preserving their evaluations, as we show now.

\begin{lemma} \label{lemma footprint reduced form}
For any monomial ordering $ \prec $ and any finite decreasing set $ \J \subseteq \mathbb{N}^m $, the maps
$$ 
\begin{array}{ccccc}
\langle \Delta_\prec(I(S;\J)) \rangle_\F & \stackrel{\rho}{\longrightarrow}& \F[\xx]/I(S;\J)&\stackrel{\eta}{\longrightarrow}& (\F^t)^n \\
f & \mapsto & f + I(S;\J) & \mapsto & \left( \left( f^{(\ii)}(\aaa) \right)_{\ii \in \J} \right)_{\aaa \in S}
\end{array}
$$
are isomorphisms of vector spaces over $ \F $, where $ t = |\J| $ and $ n = |S| $. In particular, the images of the monomials in $ \Delta_\prec(I(S;\J)) $ by $ \rho $ (i.e., modulo $ I(S;\J) $) and $ \eta \circ \rho $ (i.e., their Hasse derivatives) are $ \F $-linearly independent, and
\begin{equation}
|\Delta_\prec(I(S;\J))| = |\J|\cdot |S| = tn.
\label{eq size of footprint}
\end{equation}
\end{lemma}
\begin{proof}
The fact that $ \rho $ is an isomorphism of vector spaces over $ \F $ follows from \cite[Prop. 5.3.1]{clo1}. Finally, the fact that $ \eta $ is an isomorphism of vector spaces over $ \F $ follows from the $ \F $-linearity of Hasse derivatives, the definition of $ I(S;\J) $ (for injectivity) and Corollary \ref{cor hermite interpolation multivariate} (for surjectivity).
\end{proof} 

\begin{definition} \label{def reduced form}
With notation as in Lemma \ref{lemma footprint reduced form}, given $ f \in \F[\xx] $, we define its reduced form in $ S $ and $ \J $ as 
$$ \overline{f} = \rho^{-1} (f+I(S;\J)) \in \langle \Delta_\prec(I(S;\J)) \rangle_\F . $$
\end{definition}

We deduce the following properties.

\begin{corollary} \label{cor reduced form properties}
With notation as in Lemma \ref{lemma footprint reduced form}, given $ f \in \F[\xx] $, we have
\begin{enumerate}
    \item 
$ \deg(\overline{f}) \leq \deg(f) $ if $ \prec $ is a graded monomial ordering.
    \item 
$ \overline{f}^{(\ii)}(\aaa) = f^{(\ii)}(\aaa) $, for all $ \aaa \in S $ and all $ \ii \in \J $.
\item 
$ \overline{f} = f $ if $ f \in \langle \Delta_\prec(I(S;\J)) \rangle_\F $.
\end{enumerate} 
\end{corollary}
\begin{proof}
\begin{enumerate}
    \item 
By the Euclidean division algorithm \cite[Th. 2.3.3]{clo1}, we have that
$$ f = q_1g_1 + \cdots + q_Mg_M + r, $$
for some $ q_1, \ldots, q_M \in \F[\xx] $, where $ \mathcal{F} = \{g_1, g_2, \ldots, g_M \} $ is the Gr{\"o}bner basis of $ I(S;\J) $ from Theorem \ref{T:grobnerbasis}, and where $ r $ is an $ \F $-linear combination of monomials in $ \Delta_\prec(I(S;\J)) $. Since we are using a graded monomial ordering, it must hold $ \deg(r) \leq \deg(f) $, given how the Euclidean division algorithm works. Finally, since $ \mathcal{F} $ is a Gr{\"o}bner basis of $ I(S;\J) $, then we must have $ r = \overline{f} $ by uniqueness of remainders \cite[Prop. 2.6.1]{clo1}, and we are done.
    \item 
It follows from
$$ \left( \left( \overline{f}^{(\ii)}(\aaa) \right)_{\ii \in \J} \right)_{\aaa \in S} = \eta ( \rho(\overline{f}) ) = \eta( f + I(S;\J)) = \left( \left( f^{(\ii)}(\aaa) \right)_{\ii \in \J} \right)_{\aaa \in S}. $$
\item 
It follows by the uniqueness of remainders since $ \mathcal{F} $ is a Gr{\"o}bner basis, as in Item 1.
\end{enumerate}
\end{proof}

From the previous results, we deduce the following theorem, which includes a first formula for the dimension of general multiplicity codes as in Definition \ref{def multi codes general}.

\begin{theorem} \label{th multiplicity code dimension}
Let $ \J \subseteq \mathbb{N}^m $ be a finite decreasing set and $ \prec $ any graded monomial ordering. Then 
$$ \mathcal{M}(S,\J,k) = \left\lbrace \left( \left( f^{(\ii)}(\aaa) \right)_{\ii \in \J} \right)_{\aaa \in S} : f \in \langle \Delta_\prec(I(S;\J)) \rangle_\F \cap \F[\xx]_{<k} \right\rbrace . $$
In particular, if $M^k=\{\xx^\ii : \abs{\ii}<k\}$, we have
$$ \dim_\F( \mathcal{M}(S,\J,k)) = \abs{M^k\cap \Delta_\prec(I(S;\J))}. $$
Moreover, if $ \J = \J_r $ or $ \J = \J_{\rr-\ones} $ (or in general $ \J \subseteq \J_{\rr-\ones} $), where $ \rr = r \ones $, and $ k > \sum_{j=1}^m (rs_j-1) $, then $ \mathcal{M}(S,\J,k) = \F^{tn} $, i.e., it is the total space. 
\end{theorem}
\begin{proof}
The first equation follows by combining the three items in Corollary \ref{cor reduced form properties} with Definition \ref{def multi codes general}. From there, the formula for the dimension follows by the fact that the set of Hasse derivatives $ \eta(\rho( M^k\cap \Delta_\prec(I(S;\J)))) $ is an $ \F $-linearly independent set by Lemma \ref{lemma footprint reduced form}.
%

Finally, let $ \J = \J_{\rr-\ones} $, with $ \rr = r \ones $, and $ k = \sum_{j=1}^m (rs_j-1) +1 $. Then one can easily see that $ M^k\cap \Delta_\prec(I(S;\J_{\rr-\ones})) = \Delta_\prec(I(S;\J_{\rr-\ones})) $, since $ \Delta_\prec(I(S;\J_{\rr-\ones})) $ corresponds to the box $ \prod_{j=1}^m [0,rs_j-1] $ (see Remark \ref{remark footprint of box}). Hence $\mathcal{M}(S,\J_{\rr-\ones},k)$ is the total space by Lemma \ref{lemma footprint reduced form}. Finally, if $ \J \subseteq \J_{\rr-\ones} $, we also have $ M^k\cap \Delta_\prec(I(S;\J)) = \Delta_\prec(I(S;\J)) $ and the same result holds.
\end{proof}

We may also obtain the following reduced form of multiplicity codes, meaning that degrees in each variable may be upper bounded without leaving out any codeword.

\begin{corollary} \label{cor reduced multiplicity code}
If $ \J \subseteq \mathbb{N}^m $ is a decreasing set such that $ \J \subseteq \J_{\rr - \ones} $ (for instance $ \J = \J_r $ or $ \J = \J_{\rr - \ones} $), then
$$ \mathcal{M}(S,\J,k) = \left\lbrace \left( \left( f^{(\ii)}(\aaa) \right)_{\ii \in \J} \right)_{\aaa \in S} : f \in \F[\xx]_{<k}, \deg_{x_j}(f) \leq rs_j - 1, \forall j \in [m] \right\rbrace . $$
\end{corollary}
\begin{proof}
It follows from Theorem \ref{th multiplicity code dimension}, Remark \ref{remark footprint of box} and $ \Delta_\prec(I(S;\J)) \subseteq \Delta_\prec(I(S;\J_{\rr - \ones})) $.
\end{proof}

In some cases, explicit formulas for the dimension are possible, as we show next. This formula is also new to the best of our knowledge.

\begin{corollary} \label{cor dimension J_rr}
Take $r\in \mathbb{N}$, $ k \leq m(rs-1) $, and $S_1,\dots,S_m \subseteq \F $ such that $s_1=\cdots=s_m=s$. Then
$$
\dim_\F( \mathcal{M}(S,\J_{\rr -\ones},k))=\sum_{t=0}^{k-1} \sum_{i=0}^m (-1)^i\binom{m}{i}\binom{t-irs+m-1}{t-irs}.
$$
\end{corollary}
\begin{proof}
By Theorem \ref{th multiplicity code dimension}, we just need to count the monomials in the footprint with degree $<k$. By Remark \ref{remark footprint of box}, if we look at the exponents of the monomials in the footprint, we have a hypercube of size $rs\times\cdots\times rs = (rs)^m $. Therefore, this is the same problem as computing the dimension of a Reed--Muller code over the entire affine space $ \F_q^m $, changing $q$ with $rs$, which gives the stated formula (see \cite{kasamiRM} and \cite[Prop. 5.4.7]{pellikaanlibro}). 
\end{proof}

\begin{remark}\label{r:dimcartesian}
A similar result follows for Cartesian sets with different sizes by substituting in the formula for the dimension of affine Cartesian codes (see \cite[Thm. 3.1]{hiramAffineCartesian}) the sizes of the constituent sets by $rs_i$, for $i\in [m]$.
\end{remark}

\section{Dual} \label{sec dual}
The dual of a multiplicity code is equivalent (isometric for the folded Hamming metric) to a multiplicity code in the univariate case $ m = 1 $ (see \cite[Sec. 3]{woottersListDecodingPolynomialIdeal}). The same holds for Reed--Muller codes, i.e., multiplicity codes for several variables but multiplicity $ r = 1 $, see \cite{delsarteRM,kasamiRM}. However, this is no longer the case in general for multiplicity $ r > 1 $ and $ m > 1 $ variables, as we now show. 

\begin{example}\label{ex:dualnoesmult}
Take $ q = m = r = 2 $ and $ S_1 = S_2 = \F_2 $, i.e., $ S = \F_2^2 $. For $ k = 0,1,2,3,4,5 $, the dimensions of the multiplicity codes $ \mathcal{M}(S,r,k) $ are $ 0,1,3,6,10,12 $, respectively. Since the code length over the alphabet $ \F_2 $ is $ \binom{m+r-1}{m} |S| = 12 $, the corresponding duals have dimensions $ 12, 11,9,6,2,0 $, respectively. In other words, only those of dimensions $ 0 $, $ 6 $ and $ 12 $ can have a dual that is equivalent to a multiplicity code for $ q = m = r = 2 $ and $ S = \F_2^2 $. We will come back to the multiplicity code for $ k = 4 $ in Example \ref{ex worked example}.
\end{example}

Even though the dual is no longer equivalent to a multiplicity code, in this section, we will give explicit expressions for the dual of a multiplicity code in general. We will use this expression in Proposition \ref{prop lower bound dist mult SZ} to lower bound the minimum folded distance of such duals. 

As in \cite{woottersListDecodingPolynomialIdeal}, we consider duality with respect to the usual inner product in $ \F^{tn} $.

\subsection{Dual in the case $m = 1$}\label{subsec m=1 dual}
In this subsection, we give an explicit expression for the dual of single-variable multiplicity codes. Even though this case was solved in \cite{woottersListDecodingPolynomialIdeal}, we include it for convenience of the reader. This case will be significantly simpler since $ \J_{\rr -\ones} = \J_r $ in $ m = 1 $ variable.

Recall $S=\{a_1,\ldots,a_s\} \subseteq \F$, $G(x)=\prod_{a\in S}(x-a)$ is its vanishing polynomial, and $\{h_i(x)\}_{i=1}^s$ are the indicator functions of $S$. We denote by $h_{i}^{r+1(j)}(x)$ the $j$th derivative of $h_{i}^{r+1}(x)$.
\begin{proposition}\label{25.05.07}
The dual of
\[\mathcal{M}(S,2,k) = \left\{\left(
(f\left(a_1\right),f^{(1)}\left(a_1\right)), \ldots, (f\left(a_s\right),f^{(1)}\left(a_s\right))
\right) : f \in \F_q[x]_{<k} \right\}\]
is given by
\[\mathcal{M}(S,2,k)^\perp = \left\{\left(
(g\left(a_1\right),g^{(1)}\left(a_1\right)) M_1,\ldots, (g\left(a_s\right),g^{(1)}\left(a_s\right)) M_s
\right) : g \in \F_q[x]_{<2s-k} \right\},\]
where
$M_i =
\left(G^{(1)}(a_i)\right)^{-3}
\begin{pmatrix}
-2 G^{(2)}(a_i) & G^{(1)}(a_i) \\
G^{(1)}(a_i) & 0
\end{pmatrix}
$ (which is invertible since $ G^{(1)}(a_i) \neq 0 $), for all $ i \in [s] $.
\end{proposition}
\begin{proof}
Take $f \in \F_q[x]_{<k}$ and $g \in \F_q[x]_{<2s-k}$. By Propositions \ref{prop leibniz} and \ref{25.05.01},
\begin{align*}
fg
&=\sum_{i=1}^s\left((fg)(a_i),(fg)^{(1)}(a_i)\right) \cdot \left(h_{i,0}(x),h_{i,1}(x)\right)\\
&=\sum_{i=1}^s\left((fg)(a_i),
(f^{(1)}g)(a_i) + (fg^{(1)})(a_i)\right) \cdot \left(h_{i,0}(x),h_{i,1}(x)\right)\\
&=\sum_{i=1}^s\left(f(a_i), f^{(1)}(a_i)\right) \cdot \left(g(a_i)h_{i,0}(x) + g^{(1)}(a_i)h_{i,1}(x),
g(a_i) h_{i,1}(x)\right)\\
&=\sum_{i=1}^s \left(f(a_i), f^{(1)}(a_i)\right) 
\begin{pmatrix}
h_{i,0}(x) & h_{i,1}(x) \\
h_{i,1}(x) & 0
\end{pmatrix}  \begin{pmatrix}
    g(a_i) \\
    g^{(1)}(a_i)
\end{pmatrix} ,
\end{align*}
where, by Proposition~\ref{25.05.01},
\[
\begin{pmatrix}
h_{i,0}(x)\\
h_{i,1}(x)
\end{pmatrix}
=
\begin{pmatrix}
1 & -h_{i}^{2(1)}(a_i)\\
0 & 1
\end{pmatrix}
\begin{pmatrix}
1\\
x-a_i
\end{pmatrix}
h_i^2(x).\]
We also have that $ \coef(h_{i,0}(x),x^{2s-1}) = -h_{i}^{2(1)}(a_i) \left(G^{(1)}(a_1)\right)^{-2} $ and $ \coef(h_{i,1}(x),x^{2s-1}) = \left(G^{(1)}(a_1)\right)^{-2} $. Finally, as $\deg(fg)\leq 2s-2$ and $\deg(h_{i,0})=\deg(h_{i,1})=2s-1$, the addition of the coefficients of $fg$ in $ x^{2s-1} $ is $0$. Therefore, 
\[
\sum_{i=1}^s \left(f(a_i), f^{(1)}(a_i)\right)
\left(G^{(1)}(a_i)\right)^{-2}
\begin{pmatrix}
-h_{i}^{2(1)}(a_i) & 1 \\
1 & 0
\end{pmatrix}
\begin{pmatrix}
    g(a_i) \\
    g^{(1)}(a_i)
\end{pmatrix} = 0.
\]
Thus, the result follows using Remark \ref{25.05.08} in order to express $ h_{i}^{2(1)}(a_i) $ in terms of $ G(x) $.
\end{proof}

\begin{corollary}\label{25.05.10}
Assume $\F_q=\{a_1,\ldots,a_q\}$.
The dual of 
\[\mathcal{M}(\F_q,2,k) = \left\{\left(f\left(a_1\right),f^{(1)}\left(a_1\right), \ldots, f\left(a_q\right),f^{(1)}\left(a_q\right) \right) : f \in \F_q[x]_{<k} \right\} \]
is given by
\[\mathcal{M}(\F_q,2,k)^\perp = \left\{\left(g^{(1)}\left(a_1\right),g\left(a_1\right), \ldots, g^{(1)}\left(a_q\right),g\left(a_q\right) \right) : g \in \F_q[x]_{<2q-k} \right\}. \]
\end{corollary}
\begin{proof}
The vanishing polynomial of $\F_q$ is $G(x)=x^q-x$. Then, $G^{(1)}(x)=-1$ and $2G^{(2)}(x)=0$ (independently of $q$). Thus, the result follows from Remark~\ref{25.05.08} and Proposition~\ref{25.05.07}.
\end{proof}

\begin{proposition}\label{p:dualm=1}
The dual of
\[\mathcal{M}(S,r,k) = \left\{\left(
(f\left(a_1\right),\ldots,f^{(r-1)}\left(a_1\right)), \ldots, (f\left(a_s\right),\ldots,f^{(r-1)}\left(a_s\right))
\right) : f \in \F_q[x]_{<rs-k} \right\}\]
is given by
\[\mathcal{M}(S,r,k)^\perp = \left\{\left(
(g\left(a_1\right),\ldots,g^{(r-1)}\left(a_1\right)) M_1, \ldots, (g\left(a_s\right),\ldots,g^{(r-1)}\left(a_s\right)) M_s
\right) : g \in \F_q[x]_{<rs-k} \right\}, \]
where $ M_i $ is the invertible matrix given by
\[
M_i=
\begin{pmatrix}
\lambda_{i,0} & \lambda_{i,1} & \cdots & \lambda_{i,r-2} & \lambda_{i,r-1}\\
\lambda_{i,1} & \lambda_{i,2} & \cdots & \lambda_{i,r-1} & 0 \\
\vdots&\vdots& \iddots  & \vdots&\vdots\\
\lambda_{i,r-2} & \lambda_{i,r-1} & \cdots & 0 & 0\\
\lambda_{i,r-1} & 0 & \cdots & 0 & 0
\end{pmatrix},
\]
where $ \lambda_{i,r-1} = \coef(h_i^r(x),x^{rs-1}) = (G^{(1)}(a_i))^{-r} $ and, recursively in $ n = r-2,r-3,\ldots, 1,0, $ (in decreasing order), we define $ \lambda_{i,n} = - \sum_{k=n+1}^{r-1} h_i^{r(k-n)}(a_i) \lambda_{i,k} $.
\end{proposition}
\begin{proof}
Take $f \in \F_q[x]_{<k}$ and $g \in \F_q[x]_{<rs-k}$. By Proposition \ref{prop leibniz} and Theorem~\ref{prop hermite interpolation univariate},
\begin{align*}
fg&=\sum_{i=1}^s\left((fg)(a_i),\ldots,(fg)^{(r-1)}(a_i)\right) \cdot \left(h_{i,0}(x),\ldots,h_{i,r-1}(x)\right)\\
&=\sum_{i=1}^s\left((fg)(a_i),\ldots,
((fg^{(r-1)})(a_i) + \cdots + (f^{(r-1)}g)(a_i)\right) \cdot \left(h_{i,0}(x),\ldots,h_{i,r-1}(x)\right)\\
&=\sum_{i=1}^s \left(f(a_i),\ldots, f^{(r-1)}(a_i)\right)  
\begin{pmatrix}
h_{i,0}(x) & h_{i,1}(x) & \cdots &  h_{i,r-1}(x)\\
h_{i,1}(x) & h_{i,2}(x) & \cdots & 0\\
\vdots\\
h_{i,r-1}(x) & 0 & \cdots & 0 \\
\end{pmatrix} \begin{pmatrix}
    g(a_i) \\
    g^{(1)}(a_i) \\
    \vdots \\
    g^{(r-1)}(a_i)
\end{pmatrix},
\end{align*}
where $h_{i,0}(x),\ldots,h_{i,r-1}(x)$ are given in Theorem~\ref{prop hermite interpolation univariate}. Since $\deg(fg)\leq rs-2$, the addition of the coefficients of $fg$ in $ x^{rs-1} $ is $0$. Therefore, 
\[
\sum_{i=1}^s \left(f(a_i),\ldots, f^{(r-1)}(a_i)\right) M_i \left(g(a_i),\ldots, g^{(r-1)}(a_i)\right)^\intercal = 0,
\]
where $M_i$ is given in the proposition statement, given that $ \coef(h_{i,n}(x),x^{rs-1}) = \lambda_{i,n} $, for all $ n \in [0,r-1] $ and $ i \in [s] $. By a dimensional argument, we get the result.
\end{proof}

\begin{example}
By Proposition~\ref{p:dualm=1}, and using the Coding Theory Package~\cite{cod_package} of {\it Macaulay2}~\cite{Mac2}, or Magma~\cite{magma}, we obtain that the dual of 
\[\mathcal{M}(\F_2,3,3) =
\left\{\left(
f\left(0\right),f^{(1)}\left(0\right),f^{(2)}\left(0\right),
f\left(1\right),f^{(1)}\left(1\right),f^{(2)}\left(1\right) \right) : f \in \F_2[x]_{<3} \right\}\]
is given by
\[
\left\{\left(
(g^{(2)}+g^{(1)})\left(0\right),
(g^{(1)}+g)\left(0\right),
g\left(0\right),
(g^{(2)}+g^{(1)})\left(1\right),
(g^{(1)}+g)\left(1\right),
g\left(1\right)
\right) : g \in \F_2[x]_{<3} \right\}.\]
In contrast with the case of multiplicity $ r=1 $ (see Corollary \ref{25.05.10}), in this case of multiplicity $ r=2 $, the codewords of the form 
$$ \left(
g^{(2)}\left(0\right),g^{(1)}\left(0\right),g\left(0\right),
g^{(2)}\left(1\right),g^{(1)}\left(1\right),g\left(1\right) \right), $$
for $ g \in \F_2[x]_{<3} $, are not all in the dual $ \mathcal{M}(\F_2,3,3)^\perp $ even for $ S = \F_2 $, the whole field. For instance, if $ g =x^2 $, then
$$ \left(
g^{(2)}\left(0\right),g^{(1)}\left(0\right),g\left(0\right),
g^{(2)}\left(1\right),g^{(1)}\left(1\right),g\left(1\right) \right) = (1,0,0,1,0,1) \notin \mathcal{M}(\F_2,3,3)^\perp , $$
since this word is not orthogonal to
$$ \left(
f\left(0\right),f^{(1)}\left(0\right),f^{(2)}\left(0\right),
f\left(1\right),f^{(1)}\left(1\right),f^{(2)}\left(1\right) \right) = (0,1,0,1,1,0) \in \mathcal{M}(\F_2,3,3), $$
for $ f = x $. This is because, for $ r = 2 $, it is not true that the matrix $ M_i $ from Proposition \ref{p:dualm=1} is anti-diagonal with ones in its main anti-diagonal, even if $ S = \F_q $ is the whole field (whereas this holds for $ r = 1 $ by Corollary \ref{25.05.10}, see also Remark \ref{remark N_a simple}, Example \ref{ex worked example} and Remark \ref{remark antidiagonal for r=2}). 
\end{example}

\subsection{Dual in the case $ m = r = 2 $} \label{subsec duals m=r=2}
In this subsection, $ S = S_1 \times S_2 $ represents a  Cartesian product where every $ S_i \subseteq \F $ is finite. For convenience of the reader, we add a subsection where we compute the dual of multiplicity codes $ \mathcal{M}(S,2,k) $ for $ m = r = 2 $ and $ k \leq N = (2s_1-1) + (2s_2-1) $ (recall that for larger $ k $ such codes become the total space, see Theorem \ref{th multiplicity code dimension}). For simplicity, we will use the notation
$$ \partial_{x^iy^j} f = f^{(i,j)}(x,y) $$
for the $ (i,j)$th Hasse derivative of the bivariate polynomial $ f = f(x,y) \in \F[x,y] $. We may order the Hasse derivatives of order $ (i,j) \leq (1,1) $ of $ f $ in $ \aaa \in S $ as
$$ \left( f(\aaa), \partial_xf(\aaa), \partial_yf(\aaa), \partial_{xy}f(\aaa) \right) \in \F^4, $$
i.e., according to the graded lexicographic order with $ x \prec y $. Thus, the derivatives of total order $ i+j < 2 $ are the first three elements, i.e.,
$$ \left( f(\aaa), \partial_xf(\aaa), \partial_yf(\aaa) \right) \in \F^3. $$
Let $ \pi : \F^4 \longrightarrow \F^3 $ and $ \overline{\pi} : \F^4 \longrightarrow \F $ be the projections onto the first $ 3 $ coordinates and the last $ 1 $ coordinate, respectively, and extend them coordinatewise to $ \pi : (\F^4)^{|S|} \longrightarrow (\F^3)^{|S|} $ and $ \overline{\pi} : (\F^4)^{|S|} \longrightarrow (\F^1)^{|S|} $, respectively. Then, 
$$ \mathcal{M}(S,2,k) = \pi (\mathcal{M}(S,\J_{(1,1)},k)). $$
In particular, by \cite[Th. 1.5.7]{pless}, we have
\begin{equation}
\mathcal{M}(S,2,k)^\perp = \pi\left( \mathcal{M}(S,\J_{(1,1)},k) \right)^\perp
 = \pi(\mathcal{M}(S,\J_{(1,1)},k)^\perp \cap \ker(\overline{\pi})),
\label{eq multiplicity as shortened m=2}
\end{equation}
which is usually known as a shortened code. Hence, we first compute the dual $ \mathcal{M}(S,\J_{(1,1)},k)^\perp $.

Similarly to the evaluation function, given $f\in \F[x,y]$ and $\aaa \in S$, we define the element $\mathbf{f}(\aaa) = \left( f(\aaa), \partial_xf(\aaa), \partial_yf(\aaa), \partial_{xy}f(\aaa) \right) \in \F^4$.

\begin{proposition}\label{25.03.05}
Let $h_{1},\ldots,h_{s_1}$ (resp. $\ell_{1},\ldots,\ell_{s_2}$) be the indicator functions and $G_1$ (resp. $G_2$) the vanishing polynomials of $S_1$ (resp. $S_2$). The dual of $\mathcal{M}(S,\J_{(1,1)},k)
= \{ (\mathbf{f}(\aaa))_{\aaa \in S} : f \in \F[x,y]_{\leq k} \}
$ is
\begin{align*}
\mathcal{M}(S,\J_{(1,1)},k)^\perp & = \{ ( \mathbf{g}(\aaa) M_\aaa )_{\aaa \in S} : g \in \F[x,y]_{\leq N-k} \}\\
& = \{ ( \mathbf{g}(\aaa) M_\aaa )_{\aaa \in S} : g \in \F[x,y]_{\leq N-k}, \deg_x(g) \leq 2s_1-1, \deg_y(g) \leq 2s_2-1 \},
\end{align*}
where
$M_{\aaa}=
\left(G_1^{(1)}(a_i)
G_2^{(1)}(a_j)\right)^{-2}
\begin{pmatrix}
-\ell_{j}^{2(1)}(a_j) & 1\\
1 & 0
\end{pmatrix} \otimes
\begin{pmatrix}
-h_{i}^{2(1)}(a_i) & 1\\
1 & 0
\end{pmatrix}$ for $\aaa=(a_i,a_j)\in S$.
In particular, we have that $\mathcal{M}(S,\J_{(1,1)},k)^\perp$ is equivalent to $\mathcal{M}(S,\J_{(1,1)},N-k+1)$.
\end{proposition}
\begin{proof}
Let $ f,g \in \F[x,y] $ be such that $ \deg(f) \leq k-1 $ and $ \deg(g) \leq N-k $. By Corollary \ref{cor reduced form properties}, $ \overline{fg} $ has the same Hasse derivatives as $ fg $. Therefore, using the Leibniz rule (Proposition \ref{prop leibniz}) and the interpolation of Corollary~\ref{cor hermite interpolation multivariate}, we get
\begin{align*}
\label{eq hermite product 2 vars}
\overline{fg} &= \sum_{\aaa \in S}
\mathbf{fg}(\aaa) \cdot \left(h_{\aaa,(0,0)}, h_{\aaa,(1,0)}, h_{\aaa,(0,1)}, h_{\aaa,(1,1)} \right)\\
&= \sum_{\aaa \in S}
( (fg)(\aaa)h_{\aaa,(0,0)} + 
\partial_{x}(fg)(\aaa)h_{\aaa,(1,0)}
+ \partial_{y}(fg)(\aaa)h_{\aaa,(0,1)}
+ \partial_{xy}(fg)(\aaa)h_{\aaa,(1,1)})\nonumber.\\
&= \sum_{\aaa \in S} ( f(\aaa)g(\aaa) h_{\aaa,(0,0)} + (\partial_xf(\aaa)g(\aaa) + f(\aaa)\partial_xg(\aaa)) h_{\aaa,(1,0)} + (\partial_yf(\aaa)g(\aaa) +f(\aaa) \partial_yg(\aaa)) h_{\aaa,(0,1)}  \\
& + (\partial_{xy}f(\aaa)g(\aaa) + \partial_xf(\aaa)\partial_yg(\aaa) + \partial_yf(\aaa)\partial_xg(\aaa) + f(\aaa) \partial_{xy}g(\aaa)) h_{\aaa,(1,1)} ) \\
& = \sum_{\aaa \in S} \left( f(\aaa), \partial_xf(\aaa), \partial_yf(\aaa), \partial_{xy}f(\aaa) \right)
\begin{pmatrix}
h_{\aaa,(0,0)} & h_{\aaa,(1,0)} & h_{\aaa,(0,1)} & h_{\aaa,(1,1)} \\
h_{\aaa,(1,0)} & 0 & h_{\aaa,(1,1)} & 0 \\
h_{\aaa,(0,1)} & h_{\aaa,(1,1)} & 0 & 0 \\
h_{\aaa,(1,1)} & 0 & 0 & 0
\end{pmatrix}
\begin{pmatrix}
g(\aaa) \\
\partial_xg(\aaa) \\
\partial_yg(\aaa) \\
\partial_{xy}g(\aaa)
\end{pmatrix}\\
& = \sum_{\aaa=(a_i,a_j) \in S} \mathbf{f}(\aaa)
\begin{pmatrix}
h_{i,0}\ell_{j,0} & h_{i,1}\ell_{j,0} & h_{i,0}\ell_{j,1} &  h_{i,1}\ell_{j,1} \\
h_{i,1}\ell_{j,0} & 0 & h_{i,1}\ell_{j,1} & 0 \\
h_{i,0}\ell_{j,1} & h_{i,1}\ell_{j,1} & 0 & 0 \\
h_{i,1}\ell_{j,1} & 0 & 0 & 0
\end{pmatrix}
\mathbf{g}(\aaa)^\intercal \quad \text{(by Corollary~\ref{cor hermite interpolation multivariate})}\\
& = \sum_{\aaa=(a_i,a_j) \in S} \mathbf{f}(\aaa)
\begin{pmatrix}
\ell_{j,0} & \ell_{j,1}\\
\ell_{j,1} & 0
\end{pmatrix} \otimes
\begin{pmatrix}
h_{i,0} & h_{i,1}\\
h_{i,1} & 0
\end{pmatrix}
\mathbf{g}(\aaa)^\intercal.
\end{align*}
Note $ \overline{fg} \in \langle \Delta_\prec (I(S;\J_{(1,1)})) \rangle_\F $, i.e., $ \deg_x(\overline{fg}) \leq 2s_1-1 $ and $ \deg_y(\overline{fg}) \leq 2s_2-1 $ (see Remark \ref{remark footprint of box}). By Corollary~\ref{cor reduced form properties}, $ \deg(\overline{fg}) \leq \deg(fg) \leq N-1 < (2s_1-1)+(2s_2-1) $. Thus, $ \coef(\overline{fg}, x_1^{2s_1-1} x_2^{2s_2-1} ) = 0 $. Using Proposition~\ref{p:dualm=1} to compute the coefficients of $h_{i,0},h_{i,1} $ and $ \ell_{j,0},\ell_{j,1}$ in $ x_1^{2s_1-1} $ and $ x_2^{2s_2-1} $, respectively, we get
$0 = \sum_{\aaa \in S} \mathbf{f}(\aaa) M_\aaa \mathbf{g}(\aaa)^\intercal,$
where $M_{\aaa}$ is given in the proposition statement for $\aaa=(a_i,a_j)\in S$. By the symmetry of $ M_\aaa$, this means that 
$$ \left\lbrace \left( (g(\aaa), \partial_xg(\aaa), \partial_yg(\aaa), \partial_{xy}g(\aaa)) M_\aaa \right)_{\aaa \in S} : g \in \F[x,y]_{\leq N-k} \right\rbrace \subseteq \mathcal{M}(S,\J_{(1,1)},k)^\perp , $$
and they are equal since both have the same dimension by Theorem \ref{th multiplicity code dimension} or by Corollary \ref{cor dimension J_rr} (it is the same calculation as with classical Reed--Muller codes over the whole space $ \F_q^2 $ but replacing $ q $ by $ 2s $). 

Finally, for each $ g \in \F[x,y]_{\leq N-k} $, consider its reduced form $ \overline{g} \in \langle \Delta_\prec(I(S;\J_{(1,1)})) \rangle_\F $ with respect to any graded monomial ordering. By Remark \ref{remark footprint of box}, we have that $ \deg_x(\overline{g}) \leq 2s_1 - 1 $ and $ \deg_y(\overline{g}) \leq 2s_2 - 1 $. By Corollary \ref{cor reduced form properties}, $ \deg(\overline{g}) \leq \deg(g) \leq N-k $ and $ \partial_{x^iy^j}\overline{g}(\aaa) = \partial_{x^iy^j}g(\aaa) $ for all $ i,j \in \{ 0,1\} $ and all $ \aaa \in S $. Hence, we conclude the second equality for the dual is valid.
\end{proof}
We now compute the dual of classical multiplicity codes for $ m=r=2 $ by combining Proposition~\ref{25.03.05} and Eq.~(\ref{eq multiplicity as shortened m=2}).
\begin{corollary} \label{cor dual multi m=2}
Let $h_{1},\ldots,h_{s_1}$ (resp. $\ell_{1},\ldots,\ell_{s_2}$) be the indicator functions and $G_1$ (resp. $G_2$) the vanishing polynomials of $S_1$ (resp. $S_2$). The dual of $\mathcal{M}(S,2,k)$ is given by
\begin{align*}
\mathcal{M}(S,2,k)^\perp & = \{ ( (\partial_xg(\aaa), \partial_yg(\aaa), \partial_{xy}g(\aaa)) N_\aaa )_{\aaa \in S} : g \in \F[x,y]_{\leq N-k}, g(\aaa) = 0, \forall \aaa \in S \} \\
 & = \{ ( (\partial_xg(\aaa), \partial_yg(\aaa), \partial_{xy}g(\aaa)) N_\aaa )_{\aaa \in S} : \\
 & \phantom{= \{ ( (\partial_xg(\aaa),} g \in \F[x,y]_{\leq N-k}, \deg_x(g) \leq 2s_1-1, \deg_y(g) \leq 2s_2-1, g(\aaa) = 0, \forall \aaa \in S \},
\end{align*}
where $N_{\aaa}=
\left(G_1^{(1)}(a_i)
G_2^{(1)}(a_j)\right)^{-2}
\begin{pmatrix}
-\ell_{j}^{2(1)}(a_j) & 0 & 1\\
-h_{j}^{2(1)}(a_i) & 1 & 0\\
1 & 0 & 0
\end{pmatrix}$ for $\aaa=(a_i,a_j)\in S$.
\end{corollary}
\begin{proof}
By Eq.~(\ref{eq multiplicity as shortened m=2}), we need to compute
$ \C =\mathcal{M}(S,\J_{(1,1)},k)^\perp \cap \ker(\overline{\pi}) $
and project it onto the first $ 3 $ coordinates in each block of $ 4 $ coordinates. By Proposition~\ref{25.03.05}, $\C$ is formed by the codewords
\begin{equation}
\left( (g(\aaa), \partial_xg(\aaa), \partial_yg(\aaa), \partial_{xy}g(\aaa)) M_\aaa \right)_{\aaa \in S},
\label{eq codewords multi m=2}
\end{equation}
for $ g \in \F[x,y]_{\leq N-k} $  (respectively, $ \deg_x(g) \leq 2s_1-1 $ and $ \deg_y(g) \leq 2s_2-1 $), where the last component in every block of $ 4 $ coordinates is zero. That is, where $ g(\aaa) = 0 $ for all $ \aaa \in S $ (by the anti-triangular shape of $ M_\aaa $). Hence, the codewords in (\ref{eq codewords multi m=2}) have the form
$$ \left( (0, \partial_xg(\aaa), \partial_yg(\aaa), \partial_{xy}g(\aaa)) M_{\aaa} \right)_{\aaa \in S} = \left( \left( (\partial_xg(\aaa), \partial_yg(\aaa), \partial_{xy}g(\aaa) ) N_{\aaa} , 0 \right) \right)_{\aaa \in S}, $$
for $ g \in \F[x,y]_{\leq N-k} $ (respectively, $ \deg_x(g) \leq 2s_1-1 $ and $ \deg_y(g) \leq 2s_2-1 $) and $ g(\aaa)=0 $ for all $ \aaa \in S $. We obtain the result by projecting onto the first $ 3 $ coordinates in each block.
\end{proof}

In Corollary \ref{cor dual multi m=2}, we determined the dual of a multiplicity code as a vector space with certain equations, that is, it is determined implicitly. In the following result, we show how to regard this dual also as an evaluation code, and we explicitly compute one of its bases (i.e., one of its generator matrices or, equivalently, a parity-check matrix of the original multiplicity code).

\begin{proposition}\label{p:basedualm=2}
Using the notation from Theorem \ref{T:grobnerbasis}, let 
$$
\begin{aligned}
&A_1=G_1(x)\cdot \{x^\alpha y^\beta : 0\leq \alpha <s_1,\; 0\leq \beta <2s_2,\; \alpha+\beta \leq N-k-s_1 \},\\
&A_2=G_2(y)\cdot \{x^\alpha y^\beta : 0\leq \alpha < 2s_1,\; 0\leq \beta <s_2,\; \alpha+\beta \leq N-k-s_2 \}.
\end{aligned}
$$
Then,  
$$ 
\mathcal{M}(S,2,k)^\perp = \left\lbrace  \left( (\partial_xg(\aaa), \partial_yg(\aaa), \partial_{xy}g(\aaa)) \cdot N_\aaa \right)_{\aaa \in S} : g\in \langle A_1\cup A_2 \rangle_\F \right\rbrace.
$$
Finally, let $ \prec $ be a graded monomial ordering and let $ A \subseteq A_1 \cup A_2 $ be such that its initial monomials with respect to $ \prec $ are all different and are exactly those in
$$
 T = \{ x^\alpha y^\beta : \alpha <2s_1,\;  \beta < 2s_2,\; \alpha+\beta \leq N-k, \textrm{ and } (s_1\leq \alpha \text{ or } s_2\leq \beta) \}.
$$
Then a basis of $\mathcal{M}(S,2,k)^\perp$ is given by
$$
\{ \left( (\partial_xg(\aaa), \partial_yg(\aaa), \partial_{xy}g(\aaa)) \cdot N_\aaa \right)_{\aaa \in S} : g\in A \} .
$$
\end{proposition}
\begin{proof}
Let $ g \in \F[x,y]_{\leq N-k} $ be such that $ \deg_x(g) \leq 2s_1 - 1  $, $ \deg_y(g) \leq 2s_2 - 1 $, and $ g(\aaa) = 0 $, for all $ \aaa \in S $. By definition, we have $ g \in I(S;\J_1) $. Since $ \{ G_1(x),G_2(y) \} $ forms a Gr{\"o}bner basis of $ I(S;\J_1) $ with respect to $ \prec $ (Theorem \ref{T:grobnerbasis}), then by the Euclidean division algorithm \cite[Th. 2.3.3]{clo1}, we can write 
$$
g(x,y)=f_1(x,y)G_1(x)+f_2(x,y)G_2(y), 
$$
where $ \deg_x(f_iG_i) < 2s_1 $, $ \deg_y(f_iG_i) < 2s_2 $ and $ \deg(f_iG_i) \leq N-k $, for $ i = 1,2 $.
Hence $ g \in \langle A_1 \cup A_2 \rangle_\F $, and the first part on the generators follows. 

Next, since $ \ini_\prec(G_1(x)) = x^{s_1} $ and $ \ini_\prec(G_2(y)) = y^{s_2} $, the set of initial monomials of all the polynomials in $ A_1 \cup A_2 $ is precisely $ T $, hence we may choose $ A \subseteq A_1 \cup A_2 $ as in the proposition (notice that $ \ini_\prec(A_1) \cap \ini_\prec(A_2) \neq \varnothing $ thus necessarily $ A \subsetneq A_1 \cup A_2 $). 

Since the initial monomials of the polynomials in $ A $ are all distinct, then such polynomials are $ \F $-linearly independent, and by Lemma \ref{lemma footprint reduced form}, so are their Hasse derivatives. Since $ |A| = |T| $, then the only thing left to prove is $ |T| = \dim (\mathcal{M}(S,2,k)^\perp) $.

The map $x^\alpha y^\beta \mapsto x^{2s_1-\alpha-1}y^{2s_2-\beta-1}$ gives a bijection between $ T $ and the monomials in $ \Delta_\prec(I(S;\J_2)) $ of degree $ \geq k $. The number of such monomials is $ tn-\dim (\mathcal{M}(S,2,k)) $ by Theorem \ref{th multiplicity code dimension} (recall that $ t = |\J_2| $, $ n = |S| $ and $ tn = |\Delta_\prec(I(S;\J_2))| $ by Lemma \ref{lemma footprint reduced form}). Therefore,
$$ |T| = tn -\dim (\mathcal{M}(S,2,k))=\dim (\mathcal{M}(S,2,k)^\perp),$$ and we obtain the result.
%
\end{proof}

When we evaluate over the whole affine plane $ S = \F_q^2 $, we have the following simple formula for the matrices $ M_\aaa $ and $ N_\aaa $. 

\begin{remark} \label{remark N_a simple}
By Remark \ref{remark coef hermite r=2}, if $ S_1 = S_2 = \F_q $, i.e., $ S = \F_q^2 $, then for all $ (a,b) \in \F_q^2 $,
$$ M_{(a,b)} = \left( \begin{array}{cccc}
0 & 0 & 0 & 1 \\
0 & 0 & 1 & 0 \\
0 & 1 & 0 & 0 \\
1 & 0 & 0 & 0 
\end{array} \right) \quad \textrm{and} \quad N_{(a,b)} = \left( \begin{array}{ccc}
 0 & 0 & 1 \\
 0 & 1 & 0 \\
 1 & 0 & 0  
\end{array} \right). $$
\end{remark}

We illustrate the previous results with a worked example for $ q=m=r=2 $.

\begin{example} \label{ex worked example}
Consider $ q=m=r=2 $ and $ S_1 = S_2 = \F_2 $, i.e., $ S = \F_2^2 $, as in Example \ref{ex:dualnoesmult}. Set $ N = m(rs-1) = 6 $ and $ k = 4 $. The multiplicity code $ \mathcal{M}(S,r,k) = \mathcal{M}(\F_2^2,2,4) $ has dimension $ 10 $ and one of its bases, by Corollary \ref{cor reduced multiplicity code}, is given by the Hasse derivatives
$$ (f(a,b),\partial_x f(a,b), \partial_y f(a,b))_{(a,b) \in \F_2^2} \in (\F_2^3)^4 $$
of the monomials $ 1,x,y,x^2,xy,y^2,x^3,x^2y,xy^2,y^3 $. For instance, if $ f_1 = x $, then $ \partial_x f_1 = 1 $ and $ \partial_y f_1 = 0 $, since such Hasse derivatives coincide with classical derivatives (see Section \ref{sec preliminaries}). If we order $ S $ as $ (0,0) $, $ (1,0) $, $ (0,1) $ and $ (1,1) $, then the corresponding codeword is
$$ \mathbf{c}_1 = ((0,1,0), (1,1,0), (0,1,0), (1,1,0)) \in (\F_2^3)^4. $$
Similarly, if $ f_2 = xy $, we have $ \partial_x f_2 = y $ and $ \partial_y f_2 = x $, thus we obtain the codeword
$$ \mathbf{c}_2 = ((0,0,0), (0,0,1), (0,1,0), (1,1,1)) \in (\F_2^3)^4. $$
Now we look at the dual code $ \mathcal{M}(\F_2^2,2,4)^\perp $, which has dimension $ 12-10 = 2 $. By Proposition \ref{p:basedualm=2}, a basis of $ \mathcal{M}(\F_2^2,2,4)^\perp $ may be given by the Hasse derivatives
$$ ((\partial_x g(a,b), \partial_y g(a,b), \partial_{xy}g(a,b)) \cdot N_{(a,b)})_{(a,b) \in \F_2^2} = ( \partial_{xy}g(a,b), \partial_y g(a,b), \partial_x g(a,b))_{(a,b) \in \F_2^2} \in (\F_2^3)^4 $$
(where we have used Remark \ref{remark N_a simple} for $ N_{(a,b)} $) of the polynomials
$$ g_1 = G_1(x) = x^2+x \quad \textrm{and} \quad g_2 = G_2(y) = y^2+y . $$
Since $ \partial_xg_1 = 1 $ and $ \partial_y g_1 = \partial_{xy} g_1 = 0 $, the corresponding codeword is
$$ \mathbf{d}_1 = ((0,0,1), (0,0,1), (0,0,1), (0,0,1)) \in (\F_2^3)^4, $$
and we can easily verify that $ \mathbf{c}_1 \cdot \mathbf{d}_1 = \mathbf{c}_2 \cdot \mathbf{d}_1 = 0 $. Similarly for $ g_2 $. Using $ g_1 $ and $ g_2 $, a generator matrix of $ \mathcal{M}(\F_2^2,2,4)^\perp $, i.e., a parity-check matrix of $ \mathcal{M}(\F_2^2,2,4) $ is given by
$$ H = \left( \begin{array}{ccc|ccc|ccc|ccc}
0 & 0 & 1 & 0 & 0 & 1 & 0 & 0 & 1 & 0 & 0 & 1 \\
0 & 1 & 0 & 0 & 1 & 0 & 0 & 1 & 0 & 0 & 1 & 0
\end{array} \right) \in \F_2^{2 \times 12}. $$
\end{example}

\subsection{Dual in the general case} \label{subsec duals general}

We will first compute the dual of $ \mathcal{M}(S,\J_{\rr-\ones},k) $, where we fix $ \rr = r \ones $, for a positive integer $ r $, throughout this subsection. Recall that $ S = S_1 \times \cdots \times S_m $, where $ S_1, \ldots, S_m \subseteq \F $ are non-empty and of finite sizes $ s_i = |S_i| $, for $ i \in [m] $. 

\begin{definition} \label{def J_r ordering}
We order $ \J_{\rr - \ones} = \{ \ii_0, \ii_1, \ldots, \ii_{r^m-1} \} $ according to the graded lexicographic order. In particular, $ \ii_0 = \mathbf{0} $ and $ \ii_{r^m-1} = \rr - \ones $.
\end{definition}

The following technical lemma will be helpful when describing duals.

\begin{lemma} \label{lemma J_(m-1)(r-1)}
Ordering $ \J_{\rr - \ones} $ as in Definition \ref{def J_r ordering}, the first $ \binom{m+r-1}{m} $ elements of $ \J_{\rr - \ones} $ are exactly the $ \ii \in [0,r-1]^m $ such that $ |\ii| < r $, and the first $ r^m - \binom{m+r-1}{m} $ elements of $ \J_{\rr - \ones} $ are exactly the $ \ii \in [0,r-1]^m $ such that $ |\ii| < (m-1)(r-1) $. In other words,
\begin{equation*}
\begin{split}
\left\lbrace \ii_j : j \in \left[ 0, \binom{m+r-1}{m} - 1 \right] \right\rbrace & = \J_r , \\
\left\lbrace \ii_j : j \in \left[ 0,r^m - \binom{m+r-1}{m}-1 \right] \right\rbrace & = \J_{(m-1)(r-1)} \cap [0,r-1]^m.
\end{split}
\end{equation*}
\end{lemma}
\begin{proof}
The first equality is obvious. Now we prove the second one. Consider the map $\varphi: [0,r-1]^m\rightarrow [0,r-1]^m$ defined as
$$
\varphi(a_1,\dots,a_m)=(r-1-a_1,\dots,r-1-a_m).
$$
It is straightforward to check that this map is a bijection. Now if $(a_1,\dots,a_m)\in [0,r-1]^m$, then $\sum_{i=1}^m a_i\geq r$ if, and only if, $\sum_{i=1}^m (r-1-a_i)=m(r-1)-\sum_{i=1}^m a_i \leq m(r-1)-r = (m-1)(r-1)-1 $. Hence $\varphi([0,r-1]^m\setminus \J_r)= \J_{(m-1)(r-1)} \cap [0,r-1]^m$, and in particular,
$$
r^m-\binom{m+r-1}{m}=r^m-\abs{\J_r}=\abs{[0,r-1]^m\setminus \J_r}=\abs{ \J_{(m-1)(r-1)} \cap [0,r-1]^m},
$$
since $ \varphi $ is a bijection. Taking into account that we chose a graded ordering for $\J_{\rr - \ones}$, we have obtained that the first $r^m-\binom{m+r-1}{m}$ elements of $\J_{\rr - \ones}$ are the elements of $ \J_{(m-1)(r-1)} \cap [0,r-1]^m$, which concludes the proof. 
\end{proof}

We also need to introduce the invertible matrices that will give us the equivalence (as in Proposition \ref{prop isometries}) for the dual code.

Recall that $ S = S_1 \times \cdots \times S_m $ with $ s_i = |S_i| $ for $ i \in [m] $. Let $\{h_{S_j,a_j}\}_{s_j\in S_j}$ be the indicator functions and $G_j(x)=\prod_{a\in S_j}(x-a)$ the vanishing ideal of $S_j$.

By Corollary~\ref{cor hermite interpolation multivariate}, the rectangular Hermite interpolation is given by the polynomials
\[
\left\{h_{S,\aaa,\ii}(x) = \prod_{j=1}^m h_{S_j,a_j,i_j}(x_j)\right\}_{\aaa \in S, \hspace{0.1 cm} \ii \in \J_{\rr - \ones}},
\]
where $h_{S_j,a_j,i_j}(x_j)$ is defined in Remark~\ref{25.05.04}  (and Theorem \ref{prop hermite interpolation univariate}) for $\aaa = (a_1, \ldots, a_m) 
\in S$ and $\ii = (i_1, \ldots, i_m) \in \J_{\rr-\ones}$.
\begin{definition} \label{def matrix lambda}
For each $ \aaa \in S $ and $ \ii \in \J_{\rr-\ones} $, define 
$$ \lambda^\aaa_\ii = \coef(h_{S,\aaa, \ii},x_1^{rs_1-1}\cdots x_m^{rs_m-1}) \in \F.$$
We also set $ \lambda^\aaa_\ii = 0 $ if $ \ii \in \mathbb{N}^m \setminus \J_{\rr - \ones}$ and define the $ r^m \times r^m $ matrix $ M_\aaa = \left( \lambda^\aaa_{\ii_u + \ii_v} \right)_{u=0,v=0}^{r^m-1,r^m-1}$.
\end{definition}
We will need the following properties of such matrices.
\begin{lemma} \label{lemma matrix lambda}
$ M_\aaa $ is symmetric, upper anti-triangular ($ \lambda^\aaa_{\ii_u+\ii_v} = 0 $ if $ u+v > r^m-1 $) and every entry in its main anti-diagonal is 
$$ \lambda^\aaa_{\ii_u+\ii_v} = \lambda^\aaa_{\rr - \ones} = \coef(h_{S,\aaa, \rr - \ones},x_1^{rs_1-1}\cdots x_m^{rs_m-1}) \neq 0, $$
for $ u+v = r^m-1 $ by (\ref{eq coef of hermite basis multivariate}). In particular, $ M_\aaa $ is invertible for all $ \aaa \in S $.
\end{lemma}
\begin{proof}
First, notice that $ u+v = r^m-1 $ is equivalent to $ \ii_u + \ii_v = \rr - \ones $ due to the graded lexicographic order in Definition \ref{def J_r ordering}. For the same reason, if $ u+v > r^m -1 $, then $ \ii_u + \ii_v \in \mathbb{N}^m \setminus \J_{\rr - \ones} $, hence $ \lambda^\aaa_{\ii_u+\ii_v} = 0 $ by Definition \ref{def matrix lambda}.
\end{proof}

We may now describe the dual of $ \mathcal{M}(S,\J_{\rr - \ones},k) $ as desired.

\begin{theorem} \label{th dual box general}
Fix a positive integer $ k \leq N $, where $ N = \sum_{j=1}^m (rs_j-1) $. Then
\begin{equation*}
\begin{split}
\mathcal{M}(S,\J_{\rr - \ones},k)^\perp & = \left\lbrace \left( \left( g^{(\ii)}(\aaa) \right)_{\ii \in \J_{\rr - \ones}} \cdot M_\aaa \right)_{\aaa \in S} : g \in \F[\xx]_{\leq N-k} \right\rbrace \\
& = \left\lbrace \left( \left( g^{(\ii)}(\aaa) \right)_{\ii \in \J_{\rr - \ones}} \cdot M_\aaa \right)_{\aaa \in S} : g \in \F[\xx]_{\leq N-k}, \deg_{x_j}(g) \leq rs_j-1, \forall j \in [m] \right\rbrace ,
\end{split}
\end{equation*}
which is equivalent to $ \mathcal{M}(S,\J_{\rr - \ones}, N -k + 1) $, where $ \J_{\rr-\ones} $ is ordered according to the graded lexicographic ordering.
\end{theorem}
\begin{proof}
Let $ f,g \in \F[\xx] $ be such that $ \deg(f) \leq k-1 $ and $ \deg(g) \leq N-k $. Using Definition \ref{def J_r ordering}, we may see the codewords associated to $ f $ and $ g $ as $ (\mathbf{f}(\aaa))_{\aaa \in S} $ and $ (\mathbf{g}(\aaa))_{\aaa \in S} $, respectively, where
\begin{equation*}
\begin{split}
\mathbf{f}(\aaa) & = (f^{(\ii_0)}(\aaa),\ldots,f^{(\ii_{r^m-1})}(\aaa)), \\
\mathbf{g}(\aaa) & = (g^{(\ii_0)}(\aaa),\ldots,g^{(\ii_{r^m-1})}(\aaa)),
\end{split}
\end{equation*}
for $ \aaa \in S $. By Corollary \ref{cor reduced form properties}, $ \overline{fg} $ has the same Hasse derivatives as $ fg $, and moreover $ \overline{fg} \in \langle \Delta_\prec (I(S;\J_{\rr-\ones})) \rangle_\F $, i.e., $ \deg_{x_j}(\overline{fg}) \leq rs_j-1 $, for all $ j \in [m] $ (see Remark \ref{remark footprint of box}). Thus, by Corollary \ref{cor hermite interpolation multivariate}, we have that
\begin{equation}
\overline{fg} = \sum_{\aaa \in S} \sum_{\ii \in \J_{\rr- \ones}} (fg)^{(\ii)}(\aaa) h_{S,\aaa,\ii}.
\label{eq hermite product}
\end{equation}
Therefore, we deduce that
\begin{equation*}
\begin{split}
0 & \stackrel{(a)}{=} \coef(\overline{fg}, x_1^{rs_1-1} \cdots x_m^{rs_m-1}) \\
 & \stackrel{(b)}{=} \coef \left( \sum_{\aaa \in S} \sum_{\ii \in \J_{\rr- \ones}} (fg)^{(\ii)}(\aaa) h_{S,\aaa,\ii}, x_1^{rs_1-1} \cdots x_m^{rs_m-1} \right) \\
 & = \sum_{\aaa \in S} \sum_{\ii \in \J_{\rr- \ones}}(fg)^{(\ii)}(\aaa)  \coef( h_{S,\aaa,\ii}, x_1^{rs_1-1} \cdots x_m^{rs_m-1} ) \\
 & \stackrel{(c)}{=} \sum_{\aaa \in S} \sum_{\ii \in \J_{\rr- \ones}} \left( \sum_{\jj + \kk = \ii} f^{(\jj)}(\aaa)g^{(\kk)}(\aaa) \right) \lambda^\aaa_\ii \\
 & = \sum_{\aaa \in S} \sum_{\ii \in \J_{\rr- \ones}} \left( \sum_{\jj + \kk = \ii} f^{(\jj)}(\aaa)g^{(\kk)}(\aaa) \lambda^\aaa_{\jj+\kk} \right) \\
 & = \sum_{\aaa \in S} \sum_{\jj \in \J_{\rr- \ones}} \sum_{\kk  \in \J_{\rr- \ones}} f^{(\jj)}(\aaa)g^{(\kk)}(\aaa) \lambda^\aaa_{\jj+\kk} 
  = \sum_{\aaa \in S} \mathbf{f}(\aaa) \cdot M_\aaa \cdot \mathbf{g}(\aaa)^\intercal,
\end{split}
\end{equation*}
where we use $ \deg(\overline{fg}) \leq \deg(fg) \leq N-1 < \sum_{i=1}^m (rs_j-1) $ (by Corollary \ref{cor reduced form properties}) in (a), Eq.~(\ref{eq hermite product}) in (b), and Leibniz rule (Proposition~\ref{prop leibniz}) in (c), together with the definition of $ \lambda^\aaa_\ii $ and $ M_\aaa $. Finally, the last expression is simply
$$ 0 = \sum_{\aaa \in S} \mathbf{f}(\aaa) \cdot M_\aaa \cdot \mathbf{g}(\aaa)^\intercal = \sum_{\aaa \in S} \mathbf{f}(\aaa) \cdot (\mathbf{g}(\aaa) M_\aaa )^\intercal, $$
by the symmetry of $ M_\aaa $. This means that
$$ \left\lbrace \left( \left( g^{(\ii)}(\aaa) \right)_{\ii \in \J_{\rr - \ones}} \cdot M_\aaa \right)_{\aaa \in S} : g \in \F[\xx]_{\leq N-k} \right\rbrace \subseteq \mathcal{M}(S,\J_{\rr - \ones},k)^\perp, $$
and they are equal since both have the same dimension by Theorem \ref{th multiplicity code dimension} or by Corollary \ref{cor dimension J_rr} (similar calculation as the classical Reed--Muller codes over $ \F_q^m $ but replacing $ q $ by $ rs $).

Finally, for each $ g \in \F[\xx]_{\leq N-k} $, consider its reduced form $ \overline{g} \in \langle \Delta_\prec(I(S;\J_{\rr - \ones})) \rangle_\F $ with respect to any graded monomial ordering. By Remark \ref{remark footprint of box}, we have that $ \deg_{x_j}(\overline{g}) \leq rs_j - 1 $, for $ j \in [m] $. Moreover, by Corollary \ref{cor reduced form properties}, we have that $ \deg(\overline{g}) \leq \deg(g) \leq N-k $ and $ \overline{g}^{(\ii)}(\aaa) = g^{(\ii)}(\aaa) $ for all $ \ii \in \J_{\rr - \ones} $ and all $ \aaa \in S $. Hence, we conclude the last equality in the proposition.
\end{proof}

In order to compute the dual of $ \mathcal{M}(S,r,k) $, we need to see this code as a restriction of $ \mathcal{M}(S,\J_{\rr - \ones},k) $. 

\begin{proposition} \label{prop multi from box restriction}
Let $ \pi_{r,m} : \F^{(r^m)} \longrightarrow \F^{\binom{m+r-1}{m}} $ be given by projecting onto the first $ \binom{m+r-1}{m} $ coordinates, and extend it to $ \pi_{r,m} : (\F^{(r^m)})^{|S|} \longrightarrow (\F^{\binom{m+r-1}{m}})^{|S|} $ coordinatewise. Then
$$ \mathcal{M}(S,r,k) = \pi_{r,m}\left( \mathcal{M}(S,\J_{\rr - \ones},k) \right). $$
\end{proposition}
\begin{proof}
Straightforward using the ordering in Definition \ref{def J_r ordering}.
\end{proof}

The following lemma is well known, see \cite[Th. 1.5.7]{pless}.

\begin{lemma} \label{lemma dual restr shorten}
For any code $ \mathcal{C} \subseteq \F^n $ and any $ I \subseteq [n] $, we have
$$ (\pi_I(\mathcal{C}))^\perp = \pi_I(\mathcal{C}^\perp \cap \ker(\pi_{[n] \setminus I})), $$
where $ \pi_I : \F^n \longrightarrow \F^{|I|} $ is the projection onto the coordinates in $ I $.
\end{lemma}

\begin{definition}
Given the matrix $M_\aaa$ from Definition~\ref{def matrix lambda}, we define the matrix $ N_\aaa $ as the (invertible) square submatrix of $ M_\aaa $ formed by the entries in its last $ \binom{m+r-1}{m} $ rows and first $ \binom{m+r-1}{m} $ columns.
\end{definition}
Therefore, we conclude the following.
\begin{theorem} \label{th dual multiplicity code}
For positive integers $ r $ and $ k \leq N $, where $ N = \sum_{j=1}^m (rs_j-1) $, define
$$ \J_r^\perp = \J_{(m-1)(r-1)}\cap [0,r-1]^m = \{ \ii \in [0,r-1]^m : |\ii| < (m-1)(r-1) \} . $$
Then, we have that
\begin{equation*}
\begin{split}
\mathcal{M}(S,r,k)^\perp = \Bigg\lbrace & \left( \left( g^{(\ii)}(\aaa) \right)_{\ii \in [0,r-1]^m \setminus \J^\perp_r } \cdot N_\aaa \right)_{\aaa \in S} : \\
 & g \in \F[\xx]_{\leq N-k}, g^{(\ii)}(\aaa) = 0, \forall \ii \in \J^\perp_r, \forall \aaa \in S \Bigg\rbrace \\
 = \Bigg\lbrace & \left( \left( g^{(\ii)}(\aaa) \right)_{\ii \in [0,r-1]^m \setminus \J^\perp_r } \cdot N_\aaa \right)_{\aaa \in S} : \\
 & g \in \F[\xx]_{\leq N-k}, \deg_{x_j}(g) \leq rs_j-1, \forall j \in [m], g^{(\ii)}(\aaa) = 0, \forall \ii \in \J^\perp_r, \forall \aaa \in S \Bigg\rbrace ,
\end{split}
\end{equation*}
where $ [0,r-1]^m \setminus \J^\perp_r $ is ordered according to the graded lexicographic order.
\end{theorem}
\begin{proof}
Define $ \overline{\pi}_{r,m} : \F^{(r^m)} \longrightarrow \F^{(r^m-\binom{m+r-1}{m})} $ as the projection onto the last $ r^m-\binom{m+r-1}{m} $ coordinates and extend it blockwise $ \overline{\pi}_{r,m} : (\F^{(r^m)})^{|S|} \longrightarrow (\F^{(r^m-\binom{m+r-1}{m})})^{|S|} $. We have
\begin{align*}
\mathcal{M}(S,r,k)^\perp & \stackrel{(a)}{=} \pi_{r,m}\left( \mathcal{M}(S,\J_{\rr - \ones},k) \right)^\perp 
\stackrel{(b)}{=} \pi_{r,m}(\mathcal{M}(S,\J_{\rr - \ones},k)^\perp \cap \ker(\overline{\pi}_{r,m})),    
\end{align*}
where we use Proposition \ref{prop multi from box restriction} in (a) and Lemma \ref{lemma dual restr shorten} in (b). Hence we only need to compute $ \pi_{r,m}(\mathcal{M}(S,\J_{\rr - \ones},k)^\perp \cap \ker(\overline{\pi}_{r,m})) $. 

Take $ g \in \F[\xx]_{\leq N - k} $ (respectively, $ \deg_{x_j}(g) \leq rs_j-1 $, for $ j \in [m] $). For each $ \aaa \in S $, define the square submatrix $ P_\aaa $ of $ M_\aaa $ formed by the entries in the first $ r^m - \binom{m+r-1}{m} $ rows and last $ r^m - \binom{m+r-1}{m} $ columns. Due to the anti-triangular shape of $ M_\aaa $, we have that 
$$ \left( \left( g^{(\ii_u)}(\aaa) \right)_{u=0}^{r^m-1} \cdot M_\aaa \right)_{\aaa \in S} \in \ker(\overline{\pi}_{r,m}) 
\quad \text{if and only if} \quad 
\left( g^{(\ii_u)}(\aaa) \right)_{u=0}^{r^m - \binom{m+r-1}{m}-1} \cdot P_\aaa = \mathbf{0}, $$
for all $ \aaa \in S $. However, $ P_\aaa $ is also anti-triangular with non-zero entries in its main anti-diagonal, hence it is also invertible, and the last equation is equivalent to $ g^{(\ii_u)}(\aaa) = 0 $, for all $ u \in [0,r^m - \binom{m+r-1}{m}-1] $ and all $ \aaa \in S $. Now, if this condition is satisfied, then
$$ \pi_{r,m} \left( \left( g^{(\ii_u)}(\aaa) \right)_{u=0}^{r^m-1} \cdot M_\aaa \right) = \left( g^{(\ii_u)}(\aaa) \right)_{u=r^m - \binom{m+r-1}{m}}^{r^m-1} \cdot N_\aaa, $$
for all $ \aaa \in S $. To conclude the proof, we notice that, by Lemma \ref{lemma J_(m-1)(r-1)}, the sets $ \J^\perp_r $ and $ [0,r-1]^m \setminus \J^\perp_r $ are formed, respectively, by the first $ r^m - \binom{m+r-1}{m} $ elements and last $ \binom{m+r-1}{m} $ elements of $ \J_{\rr - \ones} $ according to the graded lexicographic order.
\end{proof}

\begin{remark} \label{remark antidiagonal for r=2}
By Remark \ref{remark coef hermite r=2}, if $ r = 2 $ and $ S_i = \F_q $ for all $ i \in [m] $, i.e., $ S = \F_q^m $, then for all $ \aaa \in \F_q^m $, we may assume that $ M_\aaa $ and $ N_\aaa $ are anti-diagonal matrices whose main anti-diagonals are formed by ones, as in Remark \ref{remark N_a simple}. This only means that the codewords of the dual $ \mathcal{M}(S,r,k)^\perp $ are of the form $ \left( \left( g^{(\ii)}(\aaa) \right)_{\ii \in [0,r-1]^m \setminus \J^\perp_r } \right)_{\aaa \in S} $, but where we order the derivatives in $ [0,r-1]^m \setminus \J^\perp_r $ in reversed order than the graded lexicographic order. See also Corollary~\ref{25.05.10} and Example \ref{ex worked example}.
\end{remark}

In the case $ r=1 $, that is, the case of Reed--Muller codes over the Cartesian product $ S = S_1 \times \cdots \times S_m $, we obtain the following well-known expression for the dual (see \cite[Theorem 2.3]{Lopez2020-mv}). We only need to note that, for $ r=1 $, $ r^m = \binom{m+r-1}{m} = 1 $.

\begin{corollary} (\textbf{Cartesian codes})
For a set $ S = S_1 \times \cdots \times S_m $, where $ S_i \subseteq \F $ is of size $ s_i $, for $ i \in [m] $, and for positive integers $ k \leq N $, where $ N = \sum_{j=1}^m (s_j-1) $, we have
\begin{equation*}
\begin{split}
\mathcal{M}(S,1,k)^\perp & = \{ ( g(\aaa) \cdot C_\aaa )_{\aaa \in S} : g \in \F[\xx]_{\leq N-k} \} \\
 & = \{ ( g(\aaa) \cdot C_\aaa )_{\aaa \in S} : g \in \F[\xx]_{\leq N-k}, \deg_{x_j}(g) \leq s_j-1, \forall j \in [m] \} ,
\end{split}
\end{equation*}
where $ C_\aaa $ can be explicitly computed in terms of the indicator functions of $S$.
\end{corollary}

Similarly to Proposition \ref{p:basedualm=2}, we can regard the dual of a multiplicity code described in Theorem \ref{th dual multiplicity code} as an evaluation code. In the next theorem, we give a basis of the dual of a multiplicity code (i.e., one of its generator matrices, or equivalently, a parity-check matrix of the original multiplicity code). 

Let $r$ and $k\leq N$ be positive integers, where $ N = \sum_{j=1}^m (rs_j-1) $. We define the sets
\begin{equation*}
\begin{split}
\mathcal{M}^{N-k}_\uu & = \left\lbrace \xx^{\ii} : \ii \in \prod_{j=1}^m[0,(r-u_j)s_j -1],\; |\ii|\leq N-k-\sum_{j=1}^m u_j s_j \right\rbrace , \\
A^k_\uu & = \prod_{j=1}^m G_j(x_j)^{u_j} \cdot \mathcal{M}^{N-k}_{\uu},
\end{split}
\end{equation*}
for $ \uu = (u_1,\dots,u_m)\in \mathcal{B}_{\J_r^\perp}$, where $G_j(x_j) = \prod_{\alpha \in S_j} (x_j-\alpha) \in \F[x_j] $ for $j \in [m]$.


\begin{theorem}\label{t:dualcomoeval}
Setting $A^k=\bigcup_{\uu \in \mathcal{B}_{\J_r^\perp}} A^k_\uu$, we have
\begin{equation}
\mathcal{M}(S,r,k)^\perp = \left\lbrace \left( \left( g^{(\ii)}(\aaa) \right)_{\ii \in [0,r-1]^m \setminus \J_r^\perp} \cdot N_\aaa \right)_{\aaa \in S} :g\in \langle A^k \rangle_\F \right\rbrace,
\label{eq basis dual general}
\end{equation}
where $ [0,r-1]^m \setminus \J_r^\perp $ is ordered according to the graded lexicographic order. Moreover, for any graded monomial ordering $ \prec $ and any subset $B^k\subseteq A^k$ with $\ini_\prec(B^k)=\ini_\prec(A^k)$ and such that $\ini_\prec(f)\neq \ini_\prec(g)$ for any two distinct $f,g\in B^k$, then $\abs{B^k}=\dim (\mathcal{M}(S,r,k)^\perp) $, and a basis of $ \mathcal{M}(S,r,k)^\perp $ is given by
\begin{equation}
\left\lbrace \left( \left( g^{(\ii)}(\aaa) \right)_{\ii \in [0,r-1]^m \setminus \J_r^\perp} \cdot N_\aaa \right)_{\aaa \in S} :g\in B^k \right\rbrace .
\label{eq basis dual general basis}
\end{equation}
\end{theorem}
\begin{proof}
We follow the idea of the proof of Proposition \ref{p:basedualm=2}. 

First, the inclusion $ \supseteq $ in (\ref{eq basis dual general}) is straightforward, since any $ g \in \langle A^k \rangle_\F $ satisfies that $ \deg(g) \leq N-k $ and $ g^{(\ii)}(\aaa) = 0 $, for all $ \ii \in \J^\perp_r $ and all $ \aaa \in S $, by definition of Hasse derivatives.

We now prove the inclusion $ \subseteq $ in (\ref{eq basis dual general}). Let $g\in \F[\xx]_{\leq N-k}$ be such that $ \deg_{x_j}(g) \leq rs_j - 1 $, for $ j \in [m] $, and $g^{(\ii)}(\aaa) = 0 $, for all $ \ii \in \J_r^\perp $ and all $ \aaa \in S $ (recall Theorem \ref{th dual multiplicity code}). In other words, $g\in I(S;\J_r^\perp) $. Since the polynomials $ \prod_{j=1}^m G_j(x_j)^{u_j} $, for $ \uu = (u_1,\ldots, u_m) \in \mathcal{B}_{\J^\perp_r} $, form a Gr{\"o}bner basis of $I(S;\J_r^\perp)$ for the ordering $ \prec $ by Theorem \ref{T:grobnerbasis}, then by the Euclidean division algorithm \cite[Th. 2.3.3]{clo1}, we can write
$$
g=\sum_{\uu \in \mathcal{B}_{\J_r^\perp}}\left( f_\uu\prod_{j=1}^mG_j(x_j)^{u_j}\right),
$$
for polynomials $f_{\uu} \in \mathbb{F}[\xx]$ such that 
$$
\deg(f_{\uu})+\sum_{j=1}^m u_j s_j = \deg \left(f_\uu \prod_{j=1}^mG_j(x_j)^{u_j} \right)  \leq \deg(g) \leq N-k, \textrm{ and}
$$
$$ \deg_{x_j}(f_\uu) + \sum_{j=1}^m u_j s_j = \deg_{x_j} \left(f_\uu \prod_{j=1}^mG_j(x_j)^{u_j} \right) \leq \deg_{x_j}(g) \leq rs_j-1, $$
for $ j \in [m] $ (recall that $ \prec $ is graded). In other words, the monomials $ \xx^\ii $ appearing in the polynomials $ f_\uu $ satisfy 
$$ |\ii| \leq N-k - \sum_{j=1}^m u_js_j \quad \textrm{and} \quad 0 \leq i_j \leq (r-u_j)s_j - 1, $$
for $ j \in [m] $. Therefore, $ g \in \langle A^k \rangle_\F $ and we have proved Eq.~(\ref{eq basis dual general}). To prove that Eq.~(\ref{eq basis dual general basis}) gives a basis of $ \mathcal{M}(S,r,k)^\perp $, we first show that $ |\ini_\prec(A^k)| = \dim(\mathcal{M}(S,r,k)^\perp) $. Define
\begin{equation*}
\begin{split}
T^k_\uu & =\{ \xx^\ii : \abs{\ii} \leq N-k \textrm{ and } u_js_j \leq i_j \leq r s_j -1, \text{ for } j \in [m] \}, \quad \textrm{for } \uu \in \mathcal{B}_{\J_r^\perp}, \\
 T^k & =\bigcup_{\uu\in \mathcal{B}_{\J_r^\perp}}T^k_\uu.
\end{split}
\end{equation*} 
By construction, $\ini_\prec(A^k)=T^k$, so we only need to show $ |T^k| = \dim(\mathcal{M}(S,r,k)^\perp) $. Clearly, the map
$$ \begin{array}{ccc}
\psi : \Delta_\prec(I(S;\J_{\rr-\ones})) & \longrightarrow &\Delta_\prec(I(S;\J_{\rr-\ones})) \\
 \xx^\ii & \mapsto & \xx^{(rs_1-1,\ldots, rs_m-1) - \ii}
\end{array} $$
is a bijection, since $ \Delta_\prec(I(S;\J_{\rr-\ones})) = \{ \xx^\ii : 0 \leq i_j \leq rs_j - 1, \forall j \in [m] \} $ by Remark \ref{remark footprint of box}. 

We next show that $ \psi $ restricts to a bijection between $ \Delta_\prec(I(S;\J_{\rr-\ones})) \setminus \Delta_\prec(I(S;\J_r^\perp)) $ and $ \Delta_\prec(I(S;\J_r)) $. Take $ \xx^\ii \in \Delta_\prec(I(S;\J_{\rr-\ones})) \setminus \Delta_\prec(I(S;\J_r^\perp)) $. By definition, we have $ \xx^\ii \in  \ini_\prec(I(S;\J^\perp_r)) $, and therefore there exists $ \uu \in \mathcal{B}_{\J^\perp_r} $ such that $ \ini_\prec(G_1(x_1)^{u_1} \cdots G_m(x_m)^{u_m}) $ divides $ \xx^\ii $, by Theorem \ref{T:grobnerbasis}. In other words, $ u_js_j \leq i_j \leq rs_j - 1 $, for all $ j \in [m] $. If $ \kk = (rs_1-1,\ldots, rs_m-1) - \ii $, then $ 0 \leq k_j \leq (r-u_j)s_j - 1 $, for all $ j \in [m] $. Now, since $ \uu \in \mathcal{B}_{\J^\perp_r} $, then by definition, $ \uu \in \J_{\rr - \ones} \setminus \J^\perp_r $, and therefore we have that $ \rr - \uu - \ones \in \J_r $, as in the proof of Lemma \ref{lemma J_(m-1)(r-1)}. Assume that $ \xx^\kk \notin \Delta_\prec(I(S;\J_r)) $. Then by definition, $ \xx^\kk \in \ini_\prec(I(S;\J_r)) $. Hence there exists $ \vv \in \mathcal{B}_{\J_r} $ such that $ \ini_\prec(G_1(x_1)^{v_1} \cdots G_m(x_m)^{v_m}) $ divides $ \xx^\kk $, by Theorem \ref{T:grobnerbasis}. Thus $ v_js_j \leq k_j \leq (r-u_j)s_j -1 $, which implies $ v_j \leq r-u_j-1 $, for all $ j \in [m] $. However this cannot happen since $ \rr - \uu - \ones \in \J_r $, $ \vv \notin \J_r $ (because $ \vv \in \mathcal{B}_{\J_r} $) and $ \J_r $ is decreasing. Thus we conclude that $ \xx^\kk = \psi(\xx^\ii) \in \Delta_\prec(I(S;\J_r)) $. Finally, we have that
$$ \abs{\Delta_\prec(I(S;\J_{\rr-\ones}))}=|\J_{\rr- \ones}|\cdot |S|= (\abs{\J_r}+\abs{\J_r^\perp})|S|=\abs{\Delta_\prec(I(S;\J_r))}+\abs{\Delta_\prec(I(S;\J_r^\perp))}, $$
by Eq.~(\ref{eq size of footprint}). Since we have shown that $ \psi (\Delta_\prec(I(S;\J_{\rr-\ones})) \setminus \Delta_\prec(I(S;\J_r^\perp))) \subseteq \Delta_\prec(I(S;\J_r))$, and they both have the same size, we have proven that
$$ \begin{array}{ccc}
\psi : \Delta_\prec(I(S;\J_{\rr-\ones})) \setminus \Delta_\prec(I(S;\J_r^\perp)) & \longrightarrow & \Delta_\prec(I(S;\J_r)) \\
 \xx^\ii & \mapsto & \xx^{(rs_1-1,\ldots, rs_m-1) - \ii}
\end{array} $$
is a bijection. Similarly to the computations above, we have that $ T^k = \{ \xx^\ii \in \Delta_\prec(I(S;\J_{\rr-\ones})) \setminus \Delta_\prec(I(S;\J_r^\perp)) : |\ii| \leq N-k \} $, and therefore it is clear that
$$
\psi(T^k) = B_{\geq k} = \{  \xx^\ii \in \Delta_\prec (I(S;\J_r)) :  \abs{\ii}\geq k \}.
$$
As a consequence, we can finally conclude that
$$ |\ini_\prec(A^k)| = |T^k| = |B_{\geq k}| = \abs{\J_r} \cdot |S| - \dim (\mathcal{M}(S,r,k))=\dim (\mathcal{M}(S,r,k)^\perp), $$
where the length of the codes is $ tn = \abs{\J_r} \cdot |S| = |\Delta_\prec (I(S;\J_r))| $ (see Eq.~(\ref{eq size of footprint})).

Coming back to the set in Eq.~(\ref{eq basis dual general basis}), we have shown that $ \abs{B^k}= |\ini_\prec(A^k)| =\dim (\mathcal{M}(S,r,k)^\perp) $. Thus, we only need to prove that the vectors
$$
 \left\{\left( \left( g^{(\ii)}(\aaa) \right)_{\ii \in [0,r-1]^m \setminus \J_r^\perp} \cdot N_\aaa \right)_{\aaa \in S} :g\in B^k \right\}
$$
are $ \F $-linearly independent. Set $ B^k = \{ g_1, \ldots, g_M \} $, with $ M = |B^k| $ and assume that there exist $ \lambda_1, \ldots, \lambda_M \in \F $ such that 
$$
\left( g^{(\ii)}(\aaa) \right)_{\ii \in [0,r-1]^m \setminus \J_r^\perp} \cdot N_\aaa=\mathbf{0},
$$
for all $\aaa\in S$, where $ g = \sum_{i=1}^M \lambda_i g_i $. Since $ N_\aaa $ is invertible (see Theorem \ref{th dual multiplicity code}), then $ g^{(\ii)}(\aaa) = 0 $, for all $ \aaa \in S $ and all $ \ii \in [0,r-1]^m \setminus \J_r^\perp $. Since the same holds for all $ \ii \in \J^\perp_r $, then $g\in I(S;\J_{\rr-\ones})$. However, if $ g \neq 0 $, then $ \ini_\prec(g) = \ini_\prec(g_j) $, for some $ j \in [M] $ (since the initials of $ g_1, \ldots , g_M $ are all distinct by hypothesis). This implies that 
$$ \ini_\prec(g) = \ini_\prec(g_j) \in \ini_\prec(B^k) \subseteq \ini_\prec(A^k) \subseteq \Delta_\prec (I(S;\J_{\rr-\ones})) , $$ 
which contradicts that $ \ini_\prec(g) \in \ini_\prec(I(S;\J_{\rr-\ones})) $. Hence, $ g = \sum_{i=1}^M \lambda_i g_i = 0 $, but then it must hold that $ \lambda_1 = \ldots = \lambda_M = 0 $ since $ g_1, \ldots, g_M $ all have distinct initial monomials.
\end{proof}

To conclude, we observe that we may give the following lower bound on the minimum distance of the dual when $ |S_1|= \cdots = |S_m| $.

\begin{proposition} \label{prop lower bound dist mult SZ dual}
If $ S = S_1 \times \cdots \times S_m $, where $ S_i \subseteq \F $, $ s = |S_1| = \ldots = |S_m| $, $ n = s^m = |S| $ and $ r $ is a positive integer, then
$$ \dd \left( \mathcal{M}(S,r,k)^\perp \right) \geq \left( 1 - \frac{m(rs-1)-k}{rs} \right) n . $$
\end{proposition}
\begin{proof}
Let $ N = \sum_{j=1}^m (r|S_j|-1) = m(rs-1) $ and $ g \in \F[\xx]_{\leq N-k} $ be such that $ g^{(\ii)}(\aaa) = 0 $, for all $ \ii \in \J^\perp_r $ and all $ \aaa \in S $. Now, if $ \aaa \in S $ is such that $ g^{(\ii)}(\aaa) = 0 $, for all $ \ii \in [0,r-1]^m \setminus \J^\perp_r $, then $ m(g,\aaa) \geq r $. Thus we deduce that the codeword associated to $ g $ in $ \mathcal{M}(S,r,k)^\perp $ (see Theorem \ref{th dual multiplicity code}) has folded weight at least $ |S| - \frac{(N-k)s^{m-1}}{r} $ by Lemma \ref{lemma SZ bound}, and the result follows.
\end{proof}

\section*{Acknowledgements}

This work was started while H. H. L{\'o}pez and U. Mart{\'i}nez-Pe\~{n}as were visiting the Simons Institute for the Theory of Computing (University of California, Berkeley), and while H. H. L{\'o}pez and E. Camps Moreno were visiting the IMUVa-Mathematics Research Institute, University of Valladolid, Spain.

This work has been partially supported by Grant PID2022-138906NB-C21 funded by MICIU/AEI/ 10.13039/501100011033 and by ERDF/EU. Hiram H. L\'opez was partially supported by the NSF grant DMS-2401558. Rodrigo San-Jos\'e was also partially supported by Grant FPU20/01311 funded by the Spanish Ministry of Universities.

\bibliographystyle{plain}

\end{document}